\documentclass[11pt,reqno]{amsart}

\usepackage[T1]{fontenc}
\usepackage{lmodern}
\usepackage{amsmath,amssymb,amsfonts,amsthm,mathtools}
\usepackage{microtype}
\usepackage{geometry}
\usepackage{booktabs,array,longtable}
\usepackage{graphicx}
\usepackage{harvard}
\usepackage{comment}
\usepackage{etoolbox}
\ifdefined\GHSreviewlines
\usepackage[mathlines]{lineno}
\fi
\usepackage[colorlinks=true,linkcolor=blue,citecolor=blue,urlcolor=blue]{hyperref}

\makeatletter
\AddToHook{begindocument/before}{%
  \def\harvardcite#1#2#3#4{%
    \global\@namedef{HAR@fn@#1}{\hyper@@link[cite]{}{cite.#1}{#2}}%
    \global\@namedef{HAR@an@#1}{\hyper@@link[cite]{}{cite.#1}{#3}}%
    \global\@namedef{HAR@yr@#1}{\hyper@@link[cite]{}{cite.#1}{#4}}%
    \global\@namedef{HAR@df@#1}{\csname HAR@fn@#1\endcsname}%
  }}
\makeatother

\theoremstyle{plain}
\newtheorem{theorem}{Theorem}[section] 
\newtheorem{lemma}{Lemma}[section]
\newtheorem{proposition}{Proposition}[section]
\newtheorem{corollary}{Corollary}[section]

\theoremstyle{definition}
\newtheorem{definition}{Definition}[section]

\theoremstyle{remark}
\newtheorem{remark}{Remark}[section]
\numberwithin{equation}{section}

\makeatletter
\let\GHSoriginalDisplayMathEnd\]
\renewcommand{\]}{%
  \GHSoriginalDisplayMathEnd
  \@afterindentfalse\@afterheading\ignorespaces}
\AfterEndEnvironment{equation}{\@afterindentfalse\@afterheading}
\makeatother

\newcommand{\R}{\mathbb R}
\newcommand{\E}{\mathcal E}
\newcommand{\dom}{\operatorname{dom}}

\DeclareMathOperator*{\argmax}{argmax}
\DeclareMathOperator*{\argmin}{argmin}

\title[Monotone comparative statics under dual submodularity]{Monotone comparative statics under
dual submodularity, with an application to aggregate deferred acceptance}

\author{Alfred Galichon}
\address{New York University, Department of Economics and Courant Institute
of Mathematical Sciences, 19 West 4th Street, New York, NY 10012}
\email{ag133@nyu.edu}

\author{Yu-Wei Hsieh}
\address{Amazon.com}
\email{yuweihsieh01@gmail.com}

\author{Maxime Sylvestre}
\address{University of Vienna, Faculty of Mathematics, Oskar-Morgenstern-Platz, 1090 Vienna, Austria}
\email{maxime.sylvestre@univie.ac.at}

\date{October 27, 2026}

\thanks{Galichon acknowledges support from NSF grant DMS-1716489 and from
the European Research Council, Grant CoG-866274. Hsieh's contribution to the
paper reflects work completed before he joined Amazon. We thank Federico
Echenique and John Quah for insightful suggestions.}

\begin{document}

\ifdefined\GHSreviewlines
\linenumbers
\let\GHSstandardlabel\label
\renewcommand{\label}[1]{\GHSstandardlabel{#1}\ifstrequal{#1}{sec:conclusion}{}{\linelabel{line:#1}}}
\fi

\begin{abstract}
We develop monotone comparative statics for constrained quantity choices
generated by convex costs satisfying dual submodularity, meaning that their
Legendre--Fenchel conjugates are submodular.
We introduce the Q-order, an exchange order for possibly set-valued quantity
choices, and prove a dual counterpart of Topkis's theorem. A closed proper
convex function is exchangeable in the Q-order if and only if its conjugate
is submodular. Exchangeability implies Q-monotonicity of rejected quantities
under price and capacity changes; when the conjugate has a sublattice domain,
the comparative-statics property is also a characterization. We apply the result
to two-sided matching with divisible masses and welfare-generated, possibly
set-valued choices. An aggregate deferred-acceptance algorithm is well
defined; the availability sequence and the least waiting-time multiplier
sequences converge; and every proposal
limit selected by the theorem yields an aggregate equilibrium matching. The construction covers classical stable nontransferable-utility
matching and aggregate random-utility specifications.
\end{abstract}

\maketitle



\bigskip


\section{Introduction}\label{sec:introduction}

Deferred acceptance is built on a monotonicity principle: once a contract is rejected, it is removed from the set of contracts that remain available. This principle is transparent in the classical model of \citeasnoun{GaleShapley}, where agents and contracts are indivisible and the algorithm records individual proposals. It is less clear how to formulate the same process in an aggregate market. If many agents share observable characteristics, proposals and rejections are no longer individual events but vectors of masses. These masses may be divisible, optimal choices may not be unique, and a rejection in one category may have to be accompanied by adjustments in several others.

\label{resp:intro-contributions}%
This paper develops the comparative-statics structure needed for such aggregate rejection processes. We say that a convex function satisfies \emph{dual submodularity} when its Legendre--Fenchel conjugate is submodular. We introduce an order on sets of quantities, the \emph{Q-order}, and a corresponding class of convex functions, which we call \emph{exchangeable}. Our first main result is a dual counterpart of Topkis's theorem: objectives ordered in the Q-order have sets of minimizers ordered in the same sense. We then show that a closed proper convex function is exchangeable if and only if it satisfies dual submodularity. This equivalence turns a familiar substitutes condition on an indirect objective into a direct exchange property on quantities. As our main application, we use this property to construct a deferred-acceptance welfare algorithm for matching markets with divisible masses and set-valued choices. The algorithm is well defined; the availability sequence and the least waiting-time multiplier sequences converge; and every proposal limit selected by the theorem yields an aggregate equilibrium matching. The resulting existence argument covers both deterministic matching and a class of aggregate random-utility models.

\label{resp:producer-example}%
The basic comparative-statics problem can be seen in a producer's choice under capacity constraints. Let $p\in\mathbb R^{\mathcal Z}$ be a vector of output prices, let $\bar q\in\mathbb R^{\mathcal Z}$ be a vector of capacities, and let $c$ be a convex cost function. The set of optimal production plans and the associated set of unused capacities are

\[
\mathcal Q(p,\bar q)=\operatorname*{argmax}_{q\leq \bar q}
\{p^\top q-c(q)\},
\qquad
\mathcal R(p,\bar q)
=\{\bar q-q:q\in\mathcal Q(p,\bar q)\}.
\]

Two comparative-statics questions arise. First, if the price of one output increases, what happens to the quantities of the other outputs that are rejected? Second, if capacities expand, does the vector of rejected quantities increase? For indivisible goods, availability can often be represented by assigning prohibitive prices to unavailable alternatives. With divisible quantities, by contrast, a price change and a capacity change are distinct perturbations, and both must be studied directly.

When the optimizer is unique and smooth, the economic intuition is familiar. The unconstrained indirect profit
\[
c^*(p)=\sup_q\{p^\top q-c(q)\}
\]
satisfies Hotelling's lemma \cite{hotelling}, so submodularity of $c^*$ means that an increase in the price of one output weakly reduces the supply of the others. Under capacity constraints, the same substitutes logic suggests that rejected quantities should move monotonically as prices or capacities change. Smooth calculations can make this intuition plausible, but they do not cover the cases of greatest interest here: nonsmooth welfare functions, multiple optimizers, discrete choices, and mixtures of discrete and continuous quantities.

Classical monotone comparative statics does not resolve this problem. Topkis's theorem \cite{Topkis1998} orders maximizers of supermodular objectives by Veinott's strong set order \cite{Veinott1989}. Equivalently, it orders minimizers of submodular objectives. Refinements such as the quasisupermodularity and single-crossing conditions of \citeasnoun{MilgromNShannon} and the interval dominance order of \citeasnoun{quah2009comparative} weaken the restrictions placed on the objective, but they continue to compare solution sets by the strong set order, that is, through lattice operations. The constrained producer problem instead requires us to compare sets of rejected \emph{quantities} generated by a convex function whose \emph{conjugate} is submodular. These solution sets need not be sublattices, and coordinatewise comparisons between arbitrary selections can fail even when the underlying substitutes intuition is valid. The relevant operation is not to take the coordinatewise meet and join of two quantity vectors. It is to transfer part of an excess in some coordinates and, when necessary, compensate that transfer in other coordinates while keeping both vectors feasible.

\label{resp:q-order-intuition}%
The Q-order formalizes this exchange operation. Informally, if one quantity vector exceeds another in some coordinates, the order asks whether any admissible reduction of that excess can be transferred between the two vectors, possibly together with a compensating transfer in coordinates where the ranking is reversed. A set that is stable under the corresponding within-set exchanges is called a \emph{matron}. For singleton sets, the Q-order reduces to the ordinary coordinatewise order; for genuinely set-valued choices, it records monotonicity without requiring every selection to be coordinatewise comparable. Thus the Q-order plays for quantity exchanges the role that Veinott's strong set order plays for lattice operations.

Our theoretical contribution has three parts.

\begin{itemize}

\item First, we introduce the Q-order and prove a monotone comparative-statics theorem: if two objective functions are Q-ordered, their argmin sets are Q-ordered. This is the formal dual of the familiar Topkis implication from ordered objectives to ordered optimizers.

\item Second, we characterize exchangeability by convex conjugacy: for every proper closed convex function $c$,

\[
c\text{ is exchangeable}
\quad\Longleftrightarrow\quad
c^*\text{ is submodular}.
\]

\item Third, we apply the characterization to the constrained producer problem. Exchangeability implies Q-monotonicity of the rejected-quantity correspondence. The converse holds when the domain of the conjugate is a sublattice. In particular, when prices and capacities increase, the rejected quantities in coordinates whose own prices are unchanged increase in the Q-order. This result simultaneously supplies the price-induced and capacity-induced responses that motivate the analysis.

\end{itemize}

The connection with deferred acceptance is direct. Consider a two-sided market in which each side is partitioned into finitely many observable types. A feasible arrangement is described by a nonnegative matrix $\mu$, where $\mu_{xy}$ is the mass of matches between types $x$ and $y$. Each side evaluates such matrices through a welfare function. In the standard finite model, a proposer selects individual contracts from those that remain available, the receiving side tentatively retains its preferred proposals, and rejected contracts are permanently removed. In the aggregate model, the same three stages operate on matrices of masses:

\textbf{Proposal step:} the proposing side chooses welfare-maximizing quantity vectors subject to the current availability constraints.

\textbf{Disposal step:} the receiving side chooses welfare-maximizing quantity vectors subject to the proposals it received.

\textbf{Update step:} the rejected mass is subtracted from the matrix of available opportunities.

\label{resp:intro-da-link}%
For this iteration to be coherent, tightening the matrix of available opportunities must not resurrect mass that was previously rejected. This is exactly the capacity-induced comparative statics captured by the Q-order. Submodularity of each side's welfare function implies exchangeability of the associated convex conjugate, and exchangeability supplies the rejection monotonicity used in the update step. This rejection monotonicity keeps the aggregate update feasible from one round to the next and is the step on which the existence argument rests.

Our deferred-acceptance welfare algorithm includes the classical finite procedure as a special case. When each observable category contains a single individual and preferences are deterministic and tie-free, the welfare maximizations select integral vertices. Starting from binary availability, the proposal, disposal, and update vectors therefore remain binary and reproduce individual deferred acceptance. More generally, the algorithm works directly with aggregate masses of offers between categories and set-valued choices, without resolving indifferences through arbitrary tie-breaking. The same construction applies to welfare functions generated by random utility, allowing continuous masses and heterogeneous idiosyncratic preferences.

This application is related to, but distinct from, the continuous deferred-acceptance framework of \citeasnoun{AlkanGale2003}. Their analysis begins with point-valued choice functions satisfying consistency and persistence. We instead begin with welfare maximization and permit set-valued optimal choices. Persistence implies a strong form of monotonicity of rejected quantities; our Q-order isolates the weaker exchange monotonicity delivered by submodular welfare. The two approaches therefore place structure at different levels: one directly on an oracle choice function, the other on the welfare function that generates a choice correspondence. This distinction is important in aggregate economic models, where welfare functions and their convex conjugates are primitive objects and nonuniqueness is economically natural.

\label{resp:intro-dca}%
Our results also connect to discrete convex analysis, introduced systematically by \citeasnoun{Murota1998}. The specific discrete result relevant here is the $M^\natural/L^\natural$ conjugacy theorem of \citeasnoun[Theorem 3.4]{Murota2009}; the translation-submodularity inequality in equation (14) therein implies ordinary submodularity. See also \citeasnoun{murotaRd} for applications in economics.

In the continuous setting, the corresponding conjugacy theory of \citeasnoun[Theorem 1.1]{MurotaShioura2004} implies that continuous $M^\natural$-convexity entails the exchangeability considered here: conjugacy turns $M^\natural$-convexity into $L^\natural$-convexity, and $L^\natural$-convexity is stronger than submodularity. The converse does not hold. In the quadratic case, the distinction is especially transparent: exchangeability requires the inverse Hessian to have nonpositive off-diagonal entries, whereas $M^\natural$-convexity additionally imposes a diagonal-dominance restriction. Our contribution is to identify the weaker exchange structure dual to ordinary submodularity, formulate its associated set order, and show that it suffices for the rejection dynamics used in aggregate deferred acceptance.

\label{resp:intro-gross-substitutes}%
This positioning also clarifies the relation to gross substitutes. For indivisible-goods demand, gross substitutes is closely tied to $M^\natural$-concavity \cite{FujishigeYang2003}, and \citeasnoun{Baldwin2019preferences} classify substitutes and complements more broadly through demand types. Gross substitutes also underpins the equilibrium-existence analysis of \citeasnoun{gul2019walrasian} for markets with indivisible goods, no transfers, and individual or aggregate allocation constraints. In matching with contracts \cite{HatfieldMilgrom2005}, substitutability is imposed directly on choice behavior. We do not use the distinct notion of ordinal concavity or claim that it is equivalent to our condition. Our maintained assumption is cardinal and welfare based: the relevant indirect welfare function is submodular, or equivalently its conjugate is exchangeable. We show that this assumption delivers Q-monotonicity of rejected quantities for possibly set-valued choices of divisible quantities and, through it, a constructive equilibrium-existence argument. The results thereby connect the comparative-statics literature initiated by \citeasnoun{Topkis1998} and \citeasnoun{MilgromNShannon}, the discrete-convex analysis of substitutes, and the matching literature following \citeasnoun{GaleShapley}.

We formalized in Lean 4 all formal mathematical results stated and proved in the revision, including the comparative-statics results and the deferred-acceptance convergence argument. A formal-verification note before the proof appendices states the precise scope of the machine-checked results and describes the accompanying verification supplement.

The remainder of the paper is organized as follows. Section~\ref{sec:comparative-statics} introduces exchangeable functions, matrons, and the Q-order; states the dual monotone-comparative-statics results; gives examples; and explains the relationship with $M^\natural$-convexity and other exchange concepts. Section~\ref{sec:matching} introduces aggregate equilibrium matching, presents the deferred-acceptance welfare algorithm, and applies it to deterministic and random-utility environments. Appendix~\ref{app:comparative-statics} proves the comparative-statics results. Appendix~\ref{app:matching} proves the matching representation and the well-definedness, convergence, and equilibrium properties of the algorithm.

\section{Monotone comparative statics under dual submodularity}\label{sec:comparative-statics}

This section develops an order for quantity correspondences and the associated monotone comparative-statics theorem. The construction is dual to the familiar lattice approach: the P-order compares objects defined on prices and uses meet and join, whereas the Q-order compares objects defined on quantities and uses compensated transfers. We first give the definitions and main results, then work through examples, and finally locate the theory relative to Topkis's theorem, discrete convex analysis, gross substitutes, and matching choice functions.

\subsection{Definitions and main results}\label{sec:definitions-main-results}

\subsubsection{Veinott's P-order}\label{sec:p-order}

Let $\mathcal Z$ be a finite index set. Vectors in $\mathbb R^{\mathcal Z}$ are ordered coordinatewise. For $x\in\mathbb R^{\mathcal Z}$, write

\[
x^+ = x\vee 0, \qquad x^- = (-x)\vee 0,
\]
so that $x=x^+-x^-$. If $a\leq b$, the notation $[a,b]$ denotes the order interval $\{d:a\leq d\leq b\}$. For a set $D\subseteq\mathbb R^{\mathcal Z}$, its indicator is the extended-real-valued function

\[
\iota_D(x)=
\begin{cases}
0,&x\in D,\\
+\infty,&x\notin D.
\end{cases}
\]

Unless stated otherwise, functions in this section map $\mathbb R^{\mathcal Z}$ into $\mathbb R\cup\{+\infty\}$ and are proper, meaning that their effective domains are nonempty. For such a function $f$, its effective domain is $\operatorname{dom}f=\{x:f(x)<+\infty\}$ and its Legendre--Fenchel conjugate is

\[
f^*(y)=\sup_{x\in\mathbb R^{\mathcal Z}}\{x^\top y-f(x)\}.
\]

We use $\Gamma_0(\mathbb R^{\mathcal Z})$ for the class of proper closed convex functions, where closed means lower semicontinuous.

We recall the classical notion of Veinott's strong set order \cite{Veinott1989} and its extension to functions.

\begin{definition}[P-order on functions]\label{def:p-order}
Let $f,g:\mathbb R^{\mathcal Z}\to\mathbb R\cup\{+\infty\}$. We write $f\leq_P g$ if

\begin{equation}\label{eq:p-order-functions}
    f\leq_P g
\quad\Longleftrightarrow\quad
f(p\wedge p')+g(p\vee p')\leq f(p)+g(p')
\quad\text{for every }p,p'.
\end{equation}

In particular, $f\leq_P f$ is exactly submodularity of $f$; see \citeasnoun{BianBuhmannKrause2020} for an extensive treatment of submodular functions on continuous domains. For two sets $P,P'\subseteq\mathbb R^{\mathcal Z}$, the relation $P\leq_P P'$ means

\begin{equation}\label{sec2:p-order-set}
p\wedge p'\in P
\quad\text{and}\quad
p\vee p'\in P'
\qquad\text{for every }p\in P, p'\in P'.
\end{equation}

Equivalently, $P\leq_P P'$ if and only if $\iota_P\leq_P\iota_{P'}$. When $P=P'$, condition (\ref{sec2:p-order-set}) says that $P$ is a sublattice.
\end{definition}

The following is the two-objective minimization form of \citeasnoun[Theorem 2.8.1]{Topkis1998}.

\begin{proposition}[Topkis]\label{prop:topkis}
Let $f,g$ be proper extended-real-valued functions such that $f\leq_P g$. Then

\begin{equation}\label{eq:topkis}
\operatorname{argmin}f\leq_P\operatorname{argmin}g.
\end{equation}
\end{proposition}

The P-order is stable under addition: if $f\leq_P g$ and $h\leq_P h'$, then $f+h\leq_P g+h'$. Hence the same statement incorporates constraints. If $f\leq_P g$ and $P\leq_P P'$, then $f+\iota_P\leq_P g+\iota_{P'}$. Provided the two constrained objectives are proper, Proposition~\ref{prop:topkis} applies to the two constrained problems.

\subsubsection{The Q-order, exchangeability, and matrons}\label{sec:q-order}

The P-order relies on the lattice operations $\wedge$ and $\vee$. The operation relevant for quantities is instead a compensated transfer between two vectors. Take $q,q'\in\mathbb R^{\mathcal Z}$. A vector $\delta_1\in[0,(q-q')^+]$ transfers some quantity from $q$ to $q'$ in coordinates in which $q$ initially exceeds $q'$. A compensation $\delta_2\in[0,(q-q')^-]$ transfers quantity in the opposite direction, using only coordinates in which $q'$ initially exceeds $q$. The perturbed pair is

\begin{equation}\label{eq:transfer}
\widetilde q=q-\delta_1+\delta_2,
\qquad
\widetilde q'=q'+\delta_1-\delta_2.
\end{equation}

The transfer preserves the combined quantity: $\widetilde q+\widetilde q'=q+q'$. Unlike meet and join, it allows the adjustment required in one coordinate to be compensated across several other coordinates.

\begin{definition}[Q-order on functions]\label{def:q-order-functions}
Let $f,g:\mathbb R^{\mathcal Z}\to\mathbb R\cup\{+\infty\}$. We write $f\leq_Q g$ if, for every $q\in\operatorname{dom}f$, $q'\in\operatorname{dom}g$, and $\delta_1\in[0,(q-q')^+]$, there exists $\delta_2\in[0,(q-q')^-]$ such that

\begin{equation}\label{eq:q-order-functions}
f(q-\delta_1+\delta_2)+g(q'+\delta_1-\delta_2)
\leq f(q)+g(q').
\end{equation}

A function $f$ is \textbf{exchangeable} if $f\leq_Q f$.
\end{definition}

\begin{remark}[A common exchange formulation]\label{rem:common-exchange}
Given $q,q'$, call $\delta_1\in[0,(q-q')^+]$ an admissible forward transfer
and $\delta_2\in[0,(q-q')^-]$ an admissible counter-transfer. For such
transfers, define the exchange map
\begin{equation}\label{eq:exchange-map}
T_{q,q'}(\delta_1,\delta_2)
=
(q-\delta_1+\delta_2,\ q'+\delta_1-\delta_2),
\end{equation}
and write $f\oplus g$ for the separable sum defined by
$(f\oplus g)(q,q')=f(q)+g(q')$. The map $T_{q,q'}$ preserves
the componentwise sum of the two vectors. Moreover,
$T_{q,q'}((q-q')^+,0)=(q\wedge q',q\vee q')$. Consequently,
$f\leq_P g$ if and only if, for every $q,q'$,

\begin{equation}\label{eq:p-common-exchange}
(f\oplus g)\left(T_{q,q'}((q-q')^+,0)\right)
\leq(f\oplus g)(q,q').
\end{equation}

By contrast, $f\leq_Q g$ if and only if, for every $q,q'$ and every
admissible forward transfer $\delta_1$, there exists an admissible
counter-transfer $\delta_2$ such that

\begin{equation}\label{eq:q-common-exchange}
(f\oplus g)\left(T_{q,q'}(\delta_1,\delta_2)\right)
\leq(f\oplus g)(q,q').
\end{equation}

Thus the two orders use the same exchange map and separable sum but impose
different requirements on the admissible exchanges.
\end{remark}

Thus exchangeability says that any partial transfer from the relatively high coordinates of one production plan can be completed, if necessary, by a reverse transfer from its relatively low coordinates, without increasing total cost across the two plans. The term \emph{exchangeable} refers to this transfer property and is unrelated to exchangeability of random variables.

The same definition yields a relation on sets.

\begin{definition}[Q-order on sets and matrons]\label{def:q-order-sets}
For $M,M'\subseteq\mathbb R^{\mathcal Z}$, write $M\leq_Q M'$ if, for every $q\in M$, $q'\in M'$, and $\delta_1\in[0,(q-q')^+]$, there exists $\delta_2\in[0,(q-q')^-]$ such that

\begin{equation}\label{eq:q-order-sets}
q-\delta_1+\delta_2\in M,
\qquad
q'+\delta_1-\delta_2\in M'.
\end{equation}

A set $M$ is a \textbf{matron} if $M\leq_Q M$.
\end{definition}

Equivalently, $M\leq_QM'$ if and only if $\iota_M\leq_Q\iota_{M'}$. We use the name \emph{matron} to evoke the exchange structure familiar from matroids and polymatroids, while emphasizing that (\ref{eq:q-order-sets}) is a different, quantity-based object. For singleton sets, the new order agrees with the coordinatewise order:

\begin{equation}\label{eq:singleton}
\{q\}\leq_Q\{q'\}
\quad\Longleftrightarrow\quad q\leq q'.
\end{equation}

For correspondences, however, the Q-order is less restrictive than requiring every selection from the first set to be coordinatewise below every selection from the second. It records whether apparent violations of coordinatewise monotonicity can be repaired through feasible exchanges.

The reason for the letters P and Q is now apparent. The P-order is naturally expressed on prices, the intensive variables in a convex dual pair, whereas the Q-order is naturally expressed on quantities, the extensive variables. The two relations are connected exactly by convex conjugacy.

\begin{theorem}[duality of the P- and Q-orders]\label{thm:pq-duality}
Let $f,g\in\Gamma_0(\mathbb R^{\mathcal Z})$. Then

\begin{equation}\label{eq:pq-duality}
f\leq_Q g
\quad\Longleftrightarrow\quad
g^*\leq_P f^*.
\end{equation}
\end{theorem}

Notice the reversal of the two functions on the right-hand side. Relation (\ref{eq:pq-duality}) is the functional duality on which the rest of the paper is built.

The Q-order admits the same elementary minimizer argument as the P-order.

\begin{theorem}[dual Topkis theorem]\label{thm:dual-topkis}
Let $f,g$ be proper extended-real-valued functions such that $f\leq_Q g$. Then

\begin{equation}\label{eq:dual-topkis}
\operatorname{argmin}f\leq_Q\operatorname{argmin}g.
\end{equation}
\end{theorem}

More generally, suppose that
$f+\iota_M\leq_Qg+\iota_{M'}$. Provided the two constrained objectives are
proper, one has

\begin{equation}\label{eq:dual-topkis-constrained}
\operatorname*{argmin}_{q\in M}f(q)
\leq_Q
\operatorname*{argmin}_{q'\in M'}g(q').
\end{equation}

To see the argument, if either minimizer set is empty, the conclusion is immediate.
Otherwise, take minimizers $q$ and $q'$ and an admissible $\delta_1$. The Q-order supplies a $\delta_2$ for which the sum of the two objective values weakly falls. Optimality says that neither term can fall individually; both perturbed points must therefore also be minimizers. The complete proof is given in Appendix~\ref{app:comparative-statics}.

Two immediate consequences are worth recording. First, the nonempty minimizer set of an exchangeable function is a matron. Second, a proper lower-semicontinuous exchangeable function is necessarily convex; in particular, a nonempty closed matron is convex. The latter fact explains why our main results are stated in the language of closed convex functions rather than arbitrary set functions.

\subsubsection{Characterization by submodularity}\label{sec:submodularity-characterization}

Setting $f=g=c$ in Theorem~\ref{thm:pq-duality} gives the central characterization.

\begin{theorem}[exchangeability and submodularity]\label{thm:exchangeability-submodularity}
Let $c\in\Gamma_0(\mathbb R^{\mathcal Z})$. Then

\begin{equation}\label{eq:exchangeability-submodularity}
c\text{ is exchangeable}
\quad\Longleftrightarrow\quad
c^*\text{ is submodular}.
\end{equation}
\end{theorem}

We know that submodularity of $c^*$ implies supermodularity of $c$, thus exchangeability is stronger than supermodularity of the primal function. Indeed, take any $q_1,q_2$ and apply Definition~\ref{def:q-order-functions} to $q=q_1\vee q_2$, $q'=q_1\wedge q_2$, $\delta_1=(q_2-q_1)^+$, and $\delta_2=0$. This yields

\begin{equation}\label{eq:supermodularity}
c(q_1)+c(q_2)
\leq c(q_1\vee q_2)+c(q_1\wedge q_2).
\end{equation}

The converse is false in general: supermodularity of $c$ does not ensure submodularity of $c^*$.

In the smooth, strictly convex case, the submodularity of $c^*$ has a familiar matrix interpretation. Suppose that $c$ is twice continuously differentiable and $\nabla^2c(q)$ is positive definite. At $p=\nabla c(q)$, Legendre duality and the inverse function theorem yield\ifdefined\GHSreviewlines\linelabel{line:eq:inverse-hessian}\fi

\begin{equation}\label{eq:inverse-hessian}
\nabla^2c^*(p)=[\nabla^2c(q)]^{-1},
\qquad p=\nabla c(q).
\end{equation}

Thus, submodularity of $c^*$ is equivalent to the off-diagonal entries of $[\nabla^2c(q)]^{-1}$ being nonpositive. Theorem~\ref{thm:exchangeability-submodularity} shows that exchangeability of $c$ is the nonsmooth, global counterpart of this inverse-Hessian condition.

\subsubsection{Price- and capacity-induced responses}\label{sec:rejection-responses}

We now apply the preceding results to the constrained producer problem. For $p,\bar q\in\mathbb R^{\mathcal Z}$, define the optimal-production and rejected-quantity correspondences

\begin{equation}\label{eq:producer-correspondences}
\mathcal Q(p,\bar q)
=\operatorname*{argmax}_{q\leq\bar q}
\{p^\top q-c(q)\},
\qquad
\mathcal R(p,\bar q)
=\{\bar q-q:q\in\mathcal Q(p,\bar q)\}.
\end{equation}

For $I\subseteq\mathcal Z$, let $\Pi_I$ denote coordinate projection onto $\mathbb R^I$. If $p\leq p'$, write

\[
I(p,p')=\{z\in\mathcal Z:p_z=p'_z\}
\]
for the outputs whose own prices have not changed.

\begin{theorem}[monotone comparative statics of rejected quantities]\label{thm:rejection-comparative-statics}
Let $c\in\Gamma_0(\mathbb R^{\mathcal Z})$ and suppose $\operatorname{dom}c^*$ is a sublattice of $\mathbb R^{\mathcal Z}$. The following statements are equivalent:

\begin{enumerate}

\item[1.] $c$ is exchangeable;

\item[2.] $c^*$ is submodular;

\item[3.] for every $p\leq p'$ and $\bar q\leq\bar q'$ such that there exists some $q\leq\bar q$ in the domain of $c$,
\begin{equation}\label{eq:rejection-cs}
\Pi_{I(p,p')}\mathcal R(p,\bar q)
\leq_Q
\Pi_{I(p,p')}\mathcal R(p',\bar q').
\end{equation}

\end{enumerate}
\end{theorem}

Condition (\ref{eq:rejection-cs}) contains the two responses discussed in the introduction.

The feasibility qualification in item 3 ensures that the two constrained objectives are proper. It does not impose attainment or require either optimizer set to be nonempty. It is automatic in the matching application, where $0\in\operatorname{dom}G^*$ and the capacity vectors are nonnegative.

The sublattice condition on $\operatorname{dom}c^*$ is needed only to recover submodularity from the comparative-statics conclusion. Without it, the always-valid implications are

\[
c\text{ exchangeable}
\quad\Longleftrightarrow\quad
c^*\text{ submodular}
\quad\Longrightarrow\quad
(\ref{eq:rejection-cs}).
\]

\textbf{Price-induced response.} Fix a capacity $\bar q$ such that there exists some $q\leq\bar q$ in the domain of $c$, and raise some prices from $p$ to $p'$. On the coordinates whose own prices are unchanged, rejected quantities increase in the Q-order:

\begin{equation}\label{eq:price-response}
\Pi_{I(p,p')}\mathcal R(p,\bar q)
\leq_Q
\Pi_{I(p,p')}\mathcal R(p',\bar q).
\end{equation}

In the point-valued case, this is the usual substitutes conclusion: making some outputs more attractive weakly reduces production, and hence weakly raises rejection, among outputs whose prices did not change.

\textbf{Capacity-induced response.} Hold $p$ fixed, let $\bar q\leq\bar q'$, and suppose that there exists some $q\leq\bar q$ in the domain of $c$. Expanding capacity gives\ifdefined\GHSreviewlines\linelabel{line:eq:capacity-response}\fi

\begin{equation}\label{eq:capacity-response}
\mathcal R(p,\bar q)
\leq_Q
\mathcal R(p,\bar q').
\end{equation}

This is the form used by the deferred-acceptance welfare algorithm in Section~\ref{sec:matching}. Available opportunities play the role of capacities. When the availability vector is reduced after a disposal phase, (\ref{eq:capacity-response}) controls the associated rejected-quantity correspondence and prevents rejected mass from reappearing through a different selection.

If $c^*$ is differentiable, its gradient is the unconstrained supply. Submodularity means\ifdefined\GHSreviewlines\linelabel{line:eq:cross-partials}\fi

\begin{equation}\label{eq:cross-partials}
\frac{\partial^2c^*(p)}{\partial p_i\partial p_j}\leq0
\qquad(i\neq j),
\end{equation}
or, equivalently, $\nabla c^*$ is a Z-function in the terminology of \citeasnoun{Rheinboldt1974}: for each $i$, the component $(\nabla c^*)_i$ is nonincreasing in every argument $p_j$ with $j\neq i$. Theorem~\ref{thm:rejection-comparative-statics} extends this differential condition to nonsmooth functions, multiple optimizers, and simultaneous capacity changes.

\subsection{Examples}\label{sec:examples}

\subsubsection{Separable functions and boxes}

Let

\begin{equation}\label{eq:separable}
c(q)=\sum_{z\in\mathcal Z}c_z(q_z),
\end{equation}
where every $c_z:\mathbb R\to\mathbb R\cup\{+\infty\}$ is closed, proper, and convex. Then $c$ is exchangeable. One can see this directly from Definition~\ref{def:q-order-functions} by taking $\delta_2=0$: in every coordinate affected by $\delta_1$, the two scalar arguments are moved toward one another while their sum is preserved, so convexity weakly lowers the sum of their costs. Equivalently,

\[
c^*(p)=\sum_{z\in\mathcal Z}c_z^*(p_z)
\]
is modular and hence submodular.

Coordinate boxes provide the corresponding basic examples of matrons. If $\underline q\leq\bar q$, then

\begin{equation}\label{eq:box}
[\underline q,\bar q]
=\{q:\underline q\leq q\leq\bar q\}
\end{equation}
is a matron, as are lower and upper orthants. This follows either by applying the exchange definition with $\delta_2=0$ or by observing that the indicator of a box is separable.

These examples are useful as a benchmark but do not exploit compensation across coordinates. The next examples do.

\subsubsection{Entropy and multinomial choice}

Let

\[
\Delta=\left\{\mu\in\mathbb R_+^{\mathcal Z}:
\sum_{z\in\mathcal Z}\mu_z=1\right\}
\]
and define negative entropy, with the convention $0\log0=0$, by

\begin{equation}\label{eq:entropy}
c(\mu)=\sum_{z\in\mathcal Z}\mu_z\log\mu_z+\iota_\Delta(\mu).
\end{equation}

Its conjugate is the log-sum-exp function

\begin{equation}\label{eq:logsumexp}
c^*(p)=\log\left(\sum_{z\in\mathcal Z}e^{p_z}\right).
\end{equation}

Writing $s_z(p)=e^{p_z}/\sum_ke^{p_k}$, its off-diagonal second derivatives are

\begin{equation}\label{eq:logsumexp-hessian}
\frac{\partial^2c^*(p)}{\partial p_i\partial p_j}
=-s_i(p)s_j(p)\leq0
\qquad(i\neq j).
\end{equation}

Thus $c^*$ is submodular and $c$ is exchangeable by Theorem~\ref{thm:exchangeability-submodularity}. Unlike the product-domain argument in the separable example, feasibility on the simplex may require a nonzero compensating transfer $\delta_2$: any mass removed from some coordinates must be restored in others. This is exactly the type of compensation encoded by the Q-order.

The example is also the elementary random-utility case. The gradient of (\ref{eq:logsumexp}) gives multinomial-logit choice probabilities. More general social-surplus functions generated by additive random utility will be used in Section~\ref{sec:matching}.

\subsubsection{Quadratic costs}

Let

\[
c(q)=\frac12q^\top Mq,
\]
where $M$ is symmetric and positive definite. Its conjugate is

\[
c^*(p)=\frac12p^\top M^{-1}p.
\]

Theorem~\ref{thm:exchangeability-submodularity} immediately gives

\begin{equation}\label{eq:quadratic-criterion}
c\text{ is exchangeable}
\quad\Longleftrightarrow\quad
(M^{-1})_{ij}\leq0\quad\text{for all }i\neq j.
\end{equation}

Thus $M^{-1}$ must be a Stieltjes matrix: it is positive definite with nonpositive off-diagonal entries.

This example also separates exchangeability from continuous $M^\natural$-convexity. Consider\ifdefined\GHSreviewlines\linelabel{line:eq:quadratic-example}\fi

\begin{equation}\label{eq:quadratic-example}
M=
\begin{pmatrix}
5&2\\
2&1
\end{pmatrix},
\qquad
M^{-1}=
\begin{pmatrix}
1&-2\\
-2&5
\end{pmatrix}.
\end{equation}

The matrix $M$ is positive definite and $M^{-1}$ has nonpositive off-diagonal entries, so $c$ is exchangeable. But the first row of $M^{-1}$ is not diagonally dominant. The quadratic criterion for continuous $M^\natural$-convexity therefore fails. Hence exchangeability is strictly weaker than $M^\natural$-convexity even for smooth, strictly convex functions on $\mathbb R^2$.

\subsubsection{Capacity constraints in the quadratic model}

The rejected-quantity correspondence can be written explicitly in the quadratic case. Let $S=M^{-1}$, so the unconstrained supply at price $p$ is $Sp$. Introduce the multiplier $\tau\geq0$ on $q\leq\bar q$ and set

\[
\rho=\bar q-Sp.
\]

The optimal rejection $r=\bar q-q$ and the multiplier solve the linear complementarity system

\begin{equation}\label{eq:lcp}
r=S\tau+\rho,
\qquad
r\geq0,\qquad \tau\geq0,\qquad r^\top\tau=0.
\end{equation}

When $S$ is a Stieltjes matrix, Theorem~\ref{thm:rejection-comparative-statics} gives monotonicity of rejected quantities. In the strictly convex quadratic specification above, the production plan is unique, so the Q-order conclusion reduces to ordinary coordinatewise monotonicity. The value of the Q-order is that the same comparative-statics statement extends beyond the strictly convex quadratic case to settings in which optimal choices are set-valued.

\subsection{Links with existing concepts}\label{sec:links}

\subsubsection{Relation with the weak set order}\label{sec:weak-set-order}

The weak set order provides a less demanding comparison of correspondences. Following \citeasnoun{che2019weak}, for two nonempty subsets $A,B\subseteq\mathbb R^{\mathcal Z}$, write $A\leq_WB$ if every $a\in A$ is below some $b\in B$ and every $b\in B$ is above some $a\in A$.

\begin{proposition}\label{prop:orders-imply-weak-order}
Let $A$ and $B$ be nonempty subsets of $\mathbb R^{\mathcal Z}$. If $A\leq_PB$ or $A\leq_QB$, then $A\leq_WB$.
\end{proposition}

The upper half of this comparison is related to the selection property underlying the gross-substitutes condition of \citeasnoun{KelsoCrawford1982}. When some prices increase, every previously demanded bundle can be associated with a new demanded bundle that retains the goods whose prices did not change. After projection onto those goods, every old demand is therefore dominated by some new demand. This is only one half of a weak-set comparison. By contrast, Proposition~\ref{prop:orders-imply-weak-order} shows that both the P-order and the Q-order imply the two-sided weak set order. The proof is in Appendix~\ref{app:comparative-statics}.

\subsubsection{Relation with the strong set order}\label{sec:strong-set-order}

The P-order on sets is precisely Veinott's strong set order. We use the letter P only to emphasize that this order operates naturally on the price side of our convex-duality framework. Thus $A\leq_PB$ means that, for every $a\in A$ and $b\in B$, one has $a\wedge b\in A$ and $a\vee b\in B$.

Theorem~\ref{thm:pq-duality} shows that the relationship between the P-order and the Q-order is an exact conjugate duality. For closed proper convex functions $f$ and $g$,

\[
f\leq_Qg
\quad\Longleftrightarrow\quad
g^*\leq_Pf^*.
\]

Thus, the Q-order is the quantity-side dual of Veinott's strong set order, rather than a weakening or strengthening of that order on the same objects.

\subsubsection{Gross substitutes and discrete $M^\natural$-convexity}\label{sec:gross-substitutes}

The gross-substitutes condition was introduced by \citeasnoun{KelsoCrawford1982}. For a set valuation $v:\{0,1\}^{\mathcal Z}\to\mathbb R$, define the demand correspondence at prices $p\in\mathbb R^{\mathcal Z}$ by

\[
D_v(p)=\operatorname*{argmax}_{x\in\{0,1\}^{\mathcal Z}}
\{v(x)-p^\top x\}.
\]

\begin{definition}[gross substitutes]\label{def:gross-substitutes}
The valuation $v$ satisfies gross substitutes if, whenever $p\leq p'$ and $x\in D_v(p)$, there exists $x'\in D_v(p')$ such that

\begin{equation}\label{eq:gross-substitutes}
x_i\leq x_i'
\qquad\text{for every }i\in\mathcal Z\text{ such that }p_i=p_i'.
\end{equation}
\end{definition}

Thus, after some prices rise, every initially demanded bundle can be associated with a newly demanded bundle that retains every initially chosen good whose price did not change. \citeasnoun{FujishigeYang2003} show that gross substitutes is equivalent to $M^\natural$-concavity of the valuation. Equivalently, under our cost convention $c=-v$, gross substitutes corresponds to discrete $M^\natural$-convexity of $c$. \citeasnoun{PAESLEME2017294} surveys this equivalence. Discrete $M^\natural$-convexity is defined more generally on integer lattices, but the gross-substitutes formulation considered here concerns the binary domain $\{0,1\}^{\mathcal Z}$.

A useful distinction appears on the dual side. Viewing $c$ as an extended-real-valued function on $\mathbb Z^{\mathcal Z}$, discrete $M^\natural/L^\natural$ conjugacy implies that $c^*$ is locally polyhedral $L^\natural$-convex; see \citeasnoun[Theorem 3.4]{Murota2009}. Its translation-submodularity inequality, equation (14) therein, implies that
\[
(p,t)\longmapsto c^*(p+t\mathbf 1)
\]
has increasing differences between $p$ and $t$. Exchangeability requires only that $c^*$ be submodular, by Theorem~\ref{thm:exchangeability-submodularity}, and imposes no analogous increasing-differences restriction involving the common shift $t$. Discrete $M^\natural$-convexity therefore contains an additional restriction that is absent from our exchangeability condition. As we shall see in the next subsection, this common-shift restriction is inherited by the continuous version of $M^\natural$-convexity developed by \citeasnoun[Theorem 1.1]{MurotaShioura2004}.

\subsubsection{Continuous $M^\natural$-convexity}\label{sec:continuous-mnatural}

\citeasnoun[Theorem 1.1]{MurotaShioura2004} establish the corresponding conjugacy result for continuous functions. If $c$ is continuous $M^\natural$-convex, then $c^*$ is $L^\natural$-convex. Consequently, $c^*$ is submodular and $(p,t)\mapsto c^*(p+t\mathbf 1)$ has increasing differences. Theorem~\ref{thm:exchangeability-submodularity} therefore yields the following implication.

\begin{proposition}\label{prop:m-natural-exchangeable}
Every closed proper continuous $M^\natural$-convex function is exchangeable.
\end{proposition}

The converse fails. In the quadratic example (\ref{eq:quadratic-example}), $c^*$ is submodular, so $c$ is exchangeable, but the negative first row sum of $\nabla^2c^*=M^{-1}$ violates the common-shift increasing-differences condition. Hence $c$ is not continuous $M^\natural$-convex. The proof of Proposition~\ref{prop:m-natural-exchangeable} is in Appendix~\ref{app:comparative-statics}.

\subsubsection{S-convexity}\label{sec:s-convexity}

\citeasnoun{chen2022s} introduce S-convexity to study substitute structures with continuous quantities. For $d\in\mathbb R^{\mathcal Z}$, write $\operatorname{supp}^+(d)=\{i:d_i>0\}$ and $\operatorname{supp}^-(d)=\{i:d_i<0\}$. Let $e_i$ denote the $i$th unit vector and set $e_0=0$.

\begin{definition}[S-convexity]\label{def:s-convexity}
A function $c:\mathbb R^{\mathcal Z}\to\mathbb R\cup\{+\infty\}$ is S-convex if it is convex and has the following property. For every $q,q'\in\operatorname{dom}c$ and every $i\in\operatorname{supp}^+(q-q')$, there exist $j\in\operatorname{supp}^-(q-q')\cup\{0\}$ and $\alpha_0,\beta>0$ such that, for every $\alpha\in[0,\alpha_0]$,

\[
c(q)+c(q')\geq
c(q-\alpha(e_i-\beta e_j))
+c(q'+\alpha(e_i-\beta e_j)).
\]
\end{definition}

The transfer reduces a coordinate in which $q$ exceeds $q'$ and either compensates it in one coordinate in which the ranking is reversed or, when $j=0$, makes no compensating transfer. Unlike the Q-order, S-convexity requires only sufficiently small single-coordinate transfers, but the compensation must use a single coordinate. \citeasnoun[Theorem 2]{chen2022s} show that the conjugate of a lower-semicontinuous proper S-convex function is submodular. Theorem~\ref{thm:exchangeability-submodularity} therefore yields the following implication.

\begin{proposition}\label{prop:s-convex-exchangeable}
Every lower-semicontinuous proper S-convex function is exchangeable.
\end{proposition}

\citeasnoun[Theorem 3]{chen2022s} also prove the converse when $c$ is continuously differentiable and strictly convex with open effective domain. Under these regularity conditions, exchangeability and S-convexity are therefore equivalent. In particular, the quadratic example (\ref{eq:quadratic-example}) is S-convex: it separates exchangeability from continuous $M^\natural$-convexity, but not from S-convexity. Whether the converse holds without such regularity remains an open problem, so we do not have an example of a closed matron whose indicator function is not S-convex. Matrons nevertheless arise directly in our applications as nonempty optimizer sets of exchangeable objectives, including the constrained production choices of Section~\ref{sec:rejection-responses} and the aggregate matching choices of Section~\ref{sec:matching}; we do not claim that the indicator functions of these particular sets fail S-convexity. The proof of Proposition~\ref{prop:s-convex-exchangeable} is in Appendix~\ref{app:comparative-statics}.

\section{NTU matching and aggregate equilibrium}\label{sec:matching}

This section applies the comparative-statics results of Section~\ref{sec:comparative-statics} to two-sided matching. We first recall the classical notion of stable matching with nontransferable utility. In that model, a matching assigns each individual to at most one partner, and stability rules out individual deviations and blocking pairs. We then show that a stable matching admits an equilibrium representation in which agents choose preferred partners at utilities adjusted by utility burning. This representation suggests an aggregate equilibrium notion generated by general welfare functions. The latter accommodates masses of agents, heterogeneous choices within observed types, and set-valued demand. We then construct a deferred-acceptance welfare algorithm and use the Q-monotonicity of rejected quantities to prove that the algorithm is well defined and yields an equilibrium-existence result.

\subsection{Classical stable NTU matching}\label{sec:classical-ntu}

Let $\mathcal X$ and $\mathcal Y$ be finite sets of individuals on the two sides of the market. An arrangement is either a match $xy\in\mathcal X\times\mathcal Y$, an unmatched individual $x0$, or an unmatched individual $0y$. Write

\[
\mathcal A
=(\mathcal X\times\mathcal Y)
\cup(\mathcal X\times\{0\})
\cup(\{0\}\times\mathcal Y).
\]

A matching is a vector $\mu\in\{0,1\}^{\mathcal A}$ satisfying

\begin{equation}\label{eq:feasibility}
\sum_{y\in\mathcal Y}\mu_{xy}+\mu_{x0}=1
\quad\text{for every }x,
\qquad
\sum_{x\in\mathcal X}\mu_{xy}+\mu_{0y}=1
\quad\text{for every }y.
\end{equation}

If $x$ is matched with $y$, their utilities are $\alpha_{xy}$ and $\gamma_{xy}$, respectively. Their outside-option utilities are $\alpha_{x0}$ and $\gamma_{0y}$. The utilities realized under $\mu$ are

\begin{equation}\label{eq:realized-utilities}
u_x^\mu
=\sum_{y\in\mathcal Y}\mu_{xy}\alpha_{xy}
+\mu_{x0}\alpha_{x0},
\qquad
v_y^\mu
=\sum_{x\in\mathcal X}\mu_{xy}\gamma_{xy}
+\mu_{0y}\gamma_{0y}.
\end{equation}

\begin{definition}[stable NTU matching]\label{def:classical-stable-ntu}
A feasible matching $\mu$ is stable if

\begin{equation}\label{eq:ir-stability}
u_x^\mu\geq\alpha_{x0},
\qquad
v_y^\mu\geq\gamma_{0y},
\qquad
\max\{u_x^\mu-\alpha_{xy},v_y^\mu-\gamma_{xy}\}\geq0
\quad\text{for every }x,y.
\end{equation}
\end{definition}

The first two inequalities impose individual rationality. The third rules out a pair that would make both of its members strictly better off. We do not exclude ties. Thus an individual may be indifferent between several partners, and different stable matchings may assign different partners among tied alternatives. These are the standard individual-rationality and no-blocking-pair conditions of stable NTU matching \cite{GaleShapley,RothSotomayor1990}.

\subsubsection{Equilibrium representation of stable NTU matching}

It is useful to express stability as a market-clearing choice condition. For the moment all individual masses equal one, but we retain the notation $G_0$ and $H_0$ because the same functions have a direct aggregate extension. For matrices $U=(U_{xy})$ and $V=(V_{xy})$, define

\begin{equation}\label{eq:individual-welfare}
G_0(U)
=\sum_{x\in\mathcal X}
\max\left\{\alpha_{x0},\max_{y\in\mathcal Y}U_{xy}\right\},
\qquad
H_0(V)
=\sum_{y\in\mathcal Y}
\max\left\{\gamma_{0y},\max_{x\in\mathcal X}V_{xy}\right\}.
\end{equation}

In the choice conditions below, we use the same symbol $\mu$ for the matched-pair restriction $(\mu_{xy})_{(x,y)\in\mathcal X\times\mathcal Y}$. The unmatched masses are recovered from the feasibility conditions in (\ref{eq:feasibility}).

The subdifferential $\partial G_0(U)$ is the set of matching-demand matrices $\mu$ of the $\mathcal X$ side at the adjusted utilities $U$. If several alternatives maximize utility, the subdifferential contains the corresponding alternative choices and their convex combinations. The interpretation of $\partial H_0(V)$ is symmetric. When $\mu$ is binary, $\mu\in\partial G_0(U)\cap\partial H_0(V)$ says that every individual is assigned a maximizing alternative on both sides of the market.

We say that a feasible binary matching $\mu$ has an \emph{equilibrium representation} if there exist matrices $U,V$ such that

\begin{equation}\label{eq:classical-welfare-representation}
\mu\in\partial G_0(U)\cap\partial H_0(V),
\qquad
\max\{U_{xy}-\alpha_{xy},V_{xy}-\gamma_{xy}\}=0
\quad\text{for every }x,y.
\end{equation}

The second condition implies $U_{xy}\leq\alpha_{xy}$ and $V_{xy}\leq\gamma_{xy}$. Thus either side may reduce its utility from a potential match $xy$ by burning
\[
\tau^\alpha_{xy}=\alpha_{xy}-U_{xy}
\qquad\text{or}\qquad
\tau^\gamma_{xy}=\gamma_{xy}-V_{xy},
\]
\noindent but the two burned amounts $\tau^\alpha_{xy}$ and $\tau^\gamma_{xy}$ cannot both be positive. The utility relinquished by one side is not transferred to the other.

\begin{proposition}[equilibrium representation of stable NTU matching]\label{prop:classical-welfare-representation}
A feasible binary matching $\mu$ is a stable NTU matching if and only if it has an equilibrium representation satisfying (\ref{eq:classical-welfare-representation}).
\end{proposition}

The proof, which uses
\begin{equation}\label{eq:UV-from-uv}
U_{xy}=\min\{u_x^\mu,\alpha_{xy}\},
\qquad
V_{xy}=\min\{v_y^\mu,\gamma_{xy}\},
\end{equation}
is provided in Appendix~\ref{app:matching}.

Proposition~\ref{prop:classical-welfare-representation} distinguishes the two interpretations. Under classical stability, $u^\mu$ and $v^\mu$ are the utilities actually delivered by the matching. In its equilibrium representation, $U$ and $V$ are match-specific utilities net of possible burning. At those adjusted utilities, all agents choose maximizing alternatives and the common matching clears the market.

In the stable-matching interpretation, agents need not be assigned a most-preferred partner according to the primitive utilities $\alpha$ and $\gamma$. Stability only rules out a pair that would make both of its members strictly better off. In the equilibrium representation, by contrast, each agent's assigned arrangement, including the outside option when the agent is unmatched, maximizes utility at the adjusted utilities $U$ and $V$. Thus every agent receives a first choice after the match-specific utility burning has been taken into account, although that arrangement need not be a first choice under the primitive utilities.

\subsection{Aggregate equilibrium matching}\label{sec:burn-only}

We now let $\mathcal X$ and $\mathcal Y$ denote finite sets of observable types. There is a mass $n_x>0$ of agents of type $x$ and a mass $m_y>0$ of agents of type $y$. A matching is therefore a nonnegative matrix of masses rather than a binary assignment.

For homogeneous agents within each type, the individual welfare functions in (\ref{eq:individual-welfare}) become

\begin{equation}\label{eq:G0H0}
G_0(U)
=\sum_{x\in\mathcal X}n_x
\max\left\{\alpha_{x0},\max_{y\in\mathcal Y}U_{xy}\right\},
\qquad
H_0(V)
=\sum_{y\in\mathcal Y}m_y
\max\left\{\gamma_{0y},\max_{x\in\mathcal X}V_{xy}\right\}.
\end{equation}

Their subdifferentials permit the mass of a type to be split among maximizing alternatives when agents are indifferent. Both $G_0$ and $H_0$ are submodular. To see this for $G_0$, let $g_x(U)=n_x\max\{\alpha_{x0},\max_yU_{xy}\}$. For any $U,U'$, suppose without loss of generality that $g_x(U)\leq g_x(U')$. Then $g_x(U\vee U')=g_x(U')$ and $g_x(U\wedge U')\leq g_x(U)$, which is the submodular inequality for $g_x$. Summing over $x$ proves submodularity of $G_0$, and the same argument over columns proves submodularity of $H_0$. More generally, aggregate choices may be generated by welfare functions that incorporate heterogeneity, congestion, or regularization.

\subsubsection{Welfare-generated choices}\label{sec:welfare-generated-choices}

Write $\mathcal E=\mathcal X\times\mathcal Y$, $\underline{\mathbb R}=\mathbb R\cup\{-\infty\}$, and $\overline{\mathbb R}_+=[0,+\infty]$. A value $U_{xy}=-\infty$ means that match $xy$ is unavailable to the $\mathcal X$ side.

\begin{definition}[welfare function]\label{def:welfare-function}
A function
\[
G:\mathbb R^{\mathcal E}\longrightarrow\mathbb R
\]
is a welfare function if it is closed and convex and its conjugate has compact domain satisfying
\[
0\in\operatorname{dom}G^*\subseteq\mathbb R_+^{\mathcal E}.
\]
\end{definition}

Recall that the convex conjugate of $G$ is

\[
G^*(\mu)=\sup_{U\in\mathbb R^{\mathcal E}}
\left\{\sum_{xy\in\mathcal E}\mu_{xy}U_{xy}-G(U)\right\}.
\]

The compact set $\operatorname{dom}G^*$ encodes the feasible aggregate matching masses of the $\mathcal X$ side, while $G^*$ can also encode convex congestion, heterogeneity, or regularization costs. The choice correspondence generated by $G$ is defined for $U\in\underline{\mathbb R}^{\mathcal E}$ by

\begin{equation}\label{eq:choice-correspondence}
C^G(U)=
\operatorname*{argmax}_{\substack{\mu\in\operatorname{dom}G^*\\
\mu_{xy}=0\text{ if }U_{xy}=-\infty}}
\left\{
\sum_{xy:U_{xy}>-\infty}\mu_{xy}U_{xy}-G^*(\mu)
\right\}.
\end{equation}

For finite $U$, one has $C^G(U)=\partial G(U)$. More generally, fixing the set of available matches, $C^G(U)$ is the subdifferential of the welfare function obtained by restricting feasible matching masses to that set, with zero mass assigned to unavailable matches.

Define $H$ and $C^H$ symmetrically for the $\mathcal Y$ side. Any mass $\mu\in\operatorname{dom}G^*\cap\operatorname{dom}H^*$ is feasible for both sides. For the deterministic functions in (\ref{eq:G0H0}), the two domains encode the row and column population constraints, so their intersection is exactly the set of feasible matched masses.

For the comparative statics and algorithm below, we assume that $G$ and $H$ are submodular. By Theorem~\ref{thm:exchangeability-submodularity}, their conjugates $G^*$ and $H^*$ are exchangeable. This is the condition that will make their constrained rejection correspondences Q-monotone.

\subsubsection{Definition and interpretation}\label{sec:burn-only-definition}

Let $\alpha,\gamma\in\mathbb R^{\mathcal E}$ be the utilities that the two sides could obtain from each potential match before any utility burning.

\begin{definition}[aggregate equilibrium matching]\label{def:burn-only-equilibrium}
Given welfare functions $G,H$ and primitive match utilities $\alpha,\gamma$, an aggregate equilibrium matching is a triple
\[
(\mu,U,V)
\in\mathbb R_+^{\mathcal E}
\times\underline{\mathbb R}^{\mathcal E}
\times\underline{\mathbb R}^{\mathcal E}
\]
such that

\begin{equation}\label{eq:burn-only}
\mu\in C^G(U)\cap C^H(V),
\qquad
\max\{U_{xy}-\alpha_{xy},V_{xy}-\gamma_{xy}\}=0
\quad\text{for every }x,y.
\end{equation}
\end{definition}

The common choice of $\mu$ is market clearing. Each side receives a welfare-maximizing vector of matching masses at the adjusted utilities $U$ and $V$, while feasibility is encoded by the domains of $G^*$ and $H^*$. The second condition is a burn-only frontier. Equivalently, there are utility-burning matrices $\tau^\alpha,\tau^\gamma\in\overline{\mathbb R}_+^{\mathcal E}$ satisfying

\begin{equation}\label{eq:concessions}
U=\alpha-\tau^\alpha,
\qquad
V=\gamma-\tau^\gamma,
\qquad
\min\{\tau^\alpha_{xy},\tau^\gamma_{xy}\}=0
\quad\text{for every }x,y.
\end{equation}

The numeraire being burned need not be literal money: it may represent waiting \cite{Barzel1974}, costly effort, foregone amenities, or another dissipative adjustment. What matters is the absence of transfers between the two partners.

\begin{corollary}[nesting classical NTU matching]\label{cor:classical-nesting}
When all masses equal one and $\mu$ is binary, aggregate equilibrium for $G_0,H_0$ is equivalent to classical stable NTU matching.
\end{corollary}

General welfare functions extend this model to heterogeneous and possibly set-valued aggregate choices without changing the market-clearing or burn-only conditions. This is the formulation to which the deferred-acceptance welfare algorithm is applied below.

\subsubsection{Random-utility welfare}\label{sec:random-utility}

Random utility provides an important specification of Definition~\ref{def:welfare-function}. For an agent $i$ of type $x$, let the utility from matching with type $y$ be

\[
U_{xy}+\varepsilon_{iy},
\]
and let the utility from remaining unmatched be $\varepsilon_{i0}$. The distribution $\mathbf P_x$ of $(\varepsilon_{iy})_{y\in\mathcal Y\cup\{0\}}$ may depend on $x$. Define

\begin{equation}\label{eq:G-random-utility}
G(U)=\sum_{x\in\mathcal X}n_x
\mathbb E_{\mathbf P_x}
\left[
\max\left\{\varepsilon_{i0},
\max_{y\in\mathcal Y}(U_{xy}+\varepsilon_{iy})
\right\}
\right].
\end{equation}

The symmetric welfare function for the other side is

\begin{equation}\label{eq:H-random-utility}
H(V)=\sum_{y\in\mathcal Y}m_y
\mathbb E_{\mathbf Q_y}
\left[
\max\left\{\eta_{0j},
\max_{x\in\mathcal X}(V_{xy}+\eta_{xj})
\right\}
\right].
\end{equation}

Under the usual integrability condition, these functions are closed and convex, and their conjugate domains encode the population constraints \cite{GalichonSalanie2015}. They are also submodular: the maximum over alternatives is submodular in the vector of alternative-specific utilities, and expectation and summation preserve submodularity.

If ties have probability zero, the welfare function is differentiable and
aggregate demand is point-valued,

\begin{equation}\label{eq:DZW}
\mu=\nabla G(U).
\end{equation}

This is the familiar random-utility gradient representation associated with
Williams, Daly, and Zachary \cite{McFadden1981}.

Without a no-ties assumption, all demands obtainable through admissible tie-breaking are represented by

\begin{equation}\label{eq:subdifferential-demand}
\mu\in\partial G(U).
\end{equation}

The subdifferential formulation therefore permits both smooth random utility and discrete models with ties.

In the multinomial-logit case, normalize the outside-option utilities to zero and assume the random terms are independent type-I extreme-value draws. Then

\begin{equation}\label{eq:logit-welfare}
G(U)=\sum_x n_x\log\left(1+\sum_y e^{U_{xy}}\right),
\qquad
H(V)=\sum_y m_y\log\left(1+\sum_x e^{V_{xy}}\right).
\end{equation}

Writing $\mu_{x0}$ and $\mu_{0y}$ for the outside-option masses, the aggregate equilibrium conditions imply

\begin{equation}\label{eq:logit-matching}
\mu_{xy}
=\min\left\{
\mu_{x0}e^{\alpha_{xy}},
\mu_{0y}e^{\gamma_{xy}}
\right\},
\end{equation}
where the outside-option masses solve

\begin{equation}\label{eq:logit-system}
\begin{aligned}
\mu_{x0}+\sum_y
\min\{\mu_{x0}e^{\alpha_{xy}},\mu_{0y}e^{\gamma_{xy}}\}
&=n_x,\\
\mu_{0y}+\sum_x
\min\{\mu_{x0}e^{\alpha_{xy}},\mu_{0y}e^{\gamma_{xy}}\}
&=m_y.
\end{aligned}
\end{equation}

Existence for this and the more general random-utility specification follows from the algorithmic result below. Formula (\ref{eq:logit-matching}) is the nontransferable-utility case of the logit specification in \citeasnoun{GalichonKominersWeber2019}. Aggregate nontransferable-utility models with random utility are also obtained by \citeasnoun{Dagsvik2000} and, as limits of large markets, by \citeasnoun{Menzel2015}; in those models idiosyncratic preferences are specific to individual partners rather than to partner types.

\subsection{The deferred-acceptance welfare algorithm}\label{sec:algorithm}

\subsubsection{Constrained choices and waiting-time multipliers}\label{sec:constrained-choices}

Fix a submodular welfare function $G$ and initial utilities $\alpha$. Specializing the correspondences in (\ref{eq:producer-correspondences}) to $p=\alpha$, $c=G^*$, $q=\mu$, and $\bar q=\bar\mu$, write the constrained choice correspondence for a nonnegative availability vector $\bar\mu$ as

\begin{equation}\label{eq:constrained-choice}
\mathcal Q^G(\alpha,\bar\mu)
=\operatorname*{argmax}_{0\leq\mu\leq\bar\mu}
\{\mu^\top\alpha-G^*(\mu)\}.
\end{equation}

The lower bound in (\ref{eq:constrained-choice}) is automatic from
$\operatorname{dom}G^*\subseteq\mathbb R_+^{\mathcal E}$, but we display it to
emphasize that matching masses are nonnegative.

The associated rejected-quantity correspondence is

\begin{equation}\label{eq:rejected-correspondence}
\mathcal R^G(\alpha,\bar\mu)
=\{\bar\mu-\mu:\mu\in\mathcal Q^G(\alpha,\bar\mu)\}.
\end{equation}

Because $G$ is finite-valued and submodular, $G^*$ is exchangeable and
$\operatorname{dom}G=\mathbb R^{\mathcal E}$ is a sublattice, so the
qualification in Theorem~\ref{thm:rejection-comparative-statics} is satisfied. The capacity-induced response
therefore gives\ifdefined\GHSreviewlines\linelabel{line:eq:capacity-response-G}\fi

\begin{equation}\label{eq:capacity-response-G}
\bar\mu\leq\bar\mu'
\quad\Longrightarrow\quad
\mathcal R^G(\alpha,\bar\mu)
\leq_Q
\mathcal R^G(\alpha,\bar\mu').
\end{equation}

This is the precise monotonicity property used by the algorithm.

The same constrained choice has a useful dual representation. At finite
waiting times, strong duality gives the infimum formula

\begin{equation}\label{eq:finite-duality}
\max_{0\leq\mu\leq\bar\mu}
\{\mu^\top\alpha-G^*(\mu)\}
=
\inf_{\tau\in\mathbb R_+^{\mathcal E}}
\{\bar\mu^\top\tau+G(\alpha-\tau)\}.
\end{equation}

The finite infimum need not be attained when an arrangement becomes
unavailable. We therefore work on the compactified orthant
$\overline{\mathbb R}_+^{\mathcal E}=[0,+\infty]^{\mathcal E}$. If
$\tau_{xy}=+\infty$, then $\alpha_{xy}-\tau_{xy}=-\infty$, and we use the
convention $0\cdot(+\infty)=0$. Appendix~\ref{app:matching} defines the
corresponding extended pairing and the continuous extension of $G$. With these
conventions, define the attained multiplier set

\begin{equation}\label{eq:multiplier-set}
T^G(\alpha,\bar\mu)
=\operatorname*{argmin}_{\tau\in\overline{\mathbb R}_+^{\mathcal E}}
\left\{\bar\mu^\top\tau+G(\alpha-\tau)\right\}.
\end{equation}

The compactified minimum in (\ref{eq:multiplier-set}) equals the primal maximum in (\ref{eq:finite-duality}).
The exact KKT characterization, including complementary slackness at
infinite coordinates, is proved in Appendix~\ref{app:matching}. The multiplier $\tau_{xy}$
is a costly concession induced by scarcity of arrangement $xy$. In the
random-utility interpretation it can be read as a waiting cost: the relevant
systematic utility becomes $\alpha_{xy}-\tau_{xy}$.

Compactness of $\operatorname{dom}G^*$ yields a finite vector $n^G$ satisfying $\mu\leq n^G$ for every $\mu\in\operatorname{dom}G^*$. Define $n^H$ symmetrically. These side-specific bounds determine the initial availability vector and are also used when saturated coordinates of the limiting multipliers are normalized.

\subsubsection{Algorithm}\label{sec:da-algorithm}

The state variables have the same interpretation as in classical deferred acceptance:

\begin{itemize}

\item  $A$ denotes the mass of arrangements still \emph{available};

\item  $P$ denotes the mass \emph{proposed} by the $\mathcal X$ side;

\item  $T$ denotes the mass \emph{tentatively accepted} by the $\mathcal Y$ side.

\end{itemize}

Set

\begin{equation}\label{eq:da-init}
\mu^{A,0}=n^G\wedge n^H,
\qquad
\mu^{T,-1}=0.
\end{equation}

For $k=0,1,\ldots$, perform the following steps.

\textbf{Proposal phase.} Choose

\begin{equation}\label{eq:da-proposal}
\mu^{P,k}\in
\operatorname*{argmax}_{\mu^{T,k-1}\leq\mu\leq\mu^{A,k}}
\{\mu^\top\alpha-G^*(\mu)\}.
\end{equation}

The lower bound prevents a mass already tentatively accepted in the previous round from being withdrawn by the proposing side.

\textbf{Disposal phase.} Choose

\begin{equation}\label{eq:da-disposal}
\mu^{T,k}\in
\operatorname*{argmax}_{0\leq\mu\leq\mu^{P,k}}
\{\mu^\top\gamma-H^*(\mu)\}.
\end{equation}

The difference $\mu^{P,k}-\mu^{T,k}$ is the mass rejected in round $k$.

Each phase maximizes a concave objective, of the form $\mu^\top\alpha-G^*(\mu)$ or $\mu^\top\gamma-H^*(\mu)$, over a box. Equivalently, each phase is a convex optimization problem, so computational-hardness results for unrestricted submodular maximization do not apply.

\textbf{Update phase.} Remove precisely that rejected mass from the available opportunities:

\begin{equation}\label{eq:da-update}
\mu^{A,k+1}
=\mu^{A,k}-(\mu^{P,k}-\mu^{T,k}).
\end{equation}

The mathematical algorithm is the infinite sequence defined by (\ref{eq:da-init})--(\ref{eq:da-update}); it does not require a tolerance-based stopping rule. In a numerical implementation one may stop when $\|\mu^{P,k}-\mu^{T,k}\|$ falls below a chosen tolerance, but this yields an approximate fixed point and is separate from the exact convergence result below. Theorem~\ref{thm:rejection-comparative-statics} enters in the next proposal phase. Since availability has fallen from $\mu^{A,k}$ to $\mu^{A,k+1}$, Q-monotonicity of rejection and the weak-set-order implication of Proposition~\ref{prop:orders-imply-weak-order} provide an optimal proposal at the new capacity that remains above $\mu^{T,k}$. This proves that the constraint set in (\ref{eq:da-proposal}) does not become empty.

\subsubsection{Well-definedness, convergence, and existence}\label{sec:da-convergence}

For the convergence statement, choose the least optimal multipliers\ifdefined\GHSreviewlines\linelabel{line:eq:least-multipliers}\fi

\begin{equation}\label{eq:least-multipliers}
\tau^{P,k}=\inf T^G(\alpha,\mu^{A,k}),
\qquad
\tau^{T,k}=\inf T^H(\gamma,\mu^{P,k}),
\end{equation}
where the infima are coordinatewise. Submodularity makes the compactified
multiplier sets meet-closed. Compactness then implies that the coordinatewise
infimum is itself an attained multiplier, so these least elements are well
defined.

\begin{theorem}[aggregate deferred acceptance and existence]\label{thm:aggregate-da}
Let $G$ and
$H$ be submodular welfare functions. Let the side-specific vectors $n^G$ and
$n^H$ satisfy

\[
\mu\leq n^G\quad(\mu\in\operatorname{dom}G^*),
\qquad
\mu\leq n^H\quad(\mu\in\operatorname{dom}H^*),
\]
and initialize the algorithm at $\mu^{A,0}=n^G\wedge n^H$. Then:

\begin{enumerate}

\item[1.] the deferred-acceptance welfare algorithm (\ref{eq:da-init})--(\ref{eq:da-update}) is well defined;

\item[2.] the availability sequence $(\mu^{A,k})$ decreases and converges;

\item[3.] the multiplier sequence $(\tau^{P,k})$ is nondecreasing, the
multiplier sequence $(\tau^{T,k})$ is nonincreasing, and both converge in the
compactified nonnegative orthant;

\item[4.] the proposal sequence $(\mu^{P,k})$ has a convergent subsequence, and along every such subsequence the tentative-acceptance sequence has the same limit $\mu$;

\item[5.] if $(\mu,\tau^P,\tau^T)$ is the resulting joint limit and
\begin{equation}\label{eq:normalized-multipliers}
\tau_i^\alpha
=\begin{cases}
\tau_i^P,&\mu_i<n_i^G,\\
0,&\mu_i=n_i^G,
\end{cases}
\qquad
\tau_i^\gamma
=\begin{cases}
\tau_i^T,&\mu_i<n_i^H,\\
0,&\mu_i=n_i^H,
\end{cases}
\end{equation}
then
\[
(\mu,\alpha-\tau^\alpha,\gamma-\tau^\gamma)
\]
is an aggregate equilibrium matching.

\end{enumerate}
\end{theorem}

The proof is in Appendix~\ref{app:matching}. Its logic is worth summarizing. The update rule makes availability decreasing and bounded below. Moreover,

\begin{equation}\label{eq:vanishing-rejection}
\mu^{P,k}-\mu^{T,k}
=\mu^{A,k}-\mu^{A,k+1}
\longrightarrow0,
\end{equation}
so proposals and tentative acceptances have the same subsequential limits.
Compactness of the conjugate domains supplies a convergent proposal
subsequence. A closed-graph argument identifies the limiting primal choices.
Compactified duality identifies the monotone limits of the least multipliers
without taking an undefined limit of a product of a vanishing capacity and a
diverging multiplier. The resulting exact KKT conditions yield

\begin{equation}\label{eq:limit-kkt}
\mu\in C^G(\alpha-\tau^\alpha)
\cap C^H(\gamma-\tau^\gamma).
\end{equation}

Finally, complementarity of the limiting multipliers and the update rule imply

\begin{equation}\label{eq:limit-no-concession}
\min\{\tau^\alpha_{xy},\tau^\gamma_{xy}\}=0
\quad\text{for every }x,y,
\end{equation}
which is exactly the burn-only condition.

\begin{corollary}[existence]\label{cor:burn-only-existence}
Under the assumptions of Theorem~\ref{thm:aggregate-da}, an aggregate equilibrium matching exists.
\end{corollary}

In particular, the corollary applies to the deterministic welfare functions in (\ref{eq:G0H0}) and the random-utility welfare functions in (\ref{eq:G-random-utility})--(\ref{eq:H-random-utility}). Theorem~\ref{thm:aggregate-da} does not assert that the entire proposal sequence converges. Full convergence follows under additional conditions ensuring, for example, that the relevant limiting choice is unique; we do not impose such conditions here.

\subsection{Relation with the literature}\label{sec:matching-literature}

\subsubsection{Gale and Shapley}

The proposal, disposal, and update phases reproduce the economic logic of Gale and Shapley's deferred-acceptance procedure. In the classical finite model, proposals are individual contracts, the receiving side tentatively retains its preferred proposals, and each rejected contract is permanently removed. Here the same state variable records a mass of available arrangements, and the update removes the mass rejected during the current disposal phase. In a continuum model of college admissions, \citeasnoun{AzevedoLeshno2016} characterize stable matchings through market-clearing cutoffs.

When every type contains one agent and preferences are deterministic and tie-free, the welfare maximizations select integral vertices. Starting from binary availability, the proposal, disposal, and update vectors remain binary, so (\ref{eq:da-proposal})--(\ref{eq:da-update}) implement the familiar proposer-side deferred-acceptance logic. The role of Section~\ref{sec:comparative-statics} is to retain the monotonicity of permanent rejection after individual contracts are replaced by divisible masses and point-valued choices are replaced by correspondences.

The result proved here is nevertheless different in form. Classical deferred acceptance terminates after finitely many rejections because the contract set is finite. With divisible quantities, availability can fall by arbitrarily small amounts, so finite termination is not expected without additional structure. Theorem~\ref{thm:aggregate-da} instead uses compactness and monotonicity to obtain an equilibrium from a subsequential limit. It also differs from existence arguments that apply Tarski's fixed-point theorem to a monotone operator on a lattice, as in \citeasnoun{Adachi2000} and \citeasnoun{HatfieldMilgrom2005}: here choices may be set-valued and the choice sets need not be sublattices, so the argument combines the monotone convergence of availability with compactness.

\subsubsection{Alkan and Gale}\label{sec:alkan-gale}

\citeasnoun{AlkanGale2003} develop deferred acceptance for continuous schedule matching using point-valued choice functions. Given an availability vector $\bar\mu$, a choice function $C$ satisfies $C(\bar\mu)\leq\bar\mu$. Their conditions include consistency and persistence. Consistency is a revealed-preference contraction property: if the chosen vector remains feasible after the availability set contracts, it remains chosen. Persistence requires, in the notation used here,

\begin{equation}\label{eq:persistence}
C(\bar\mu')\wedge\bar\mu
\leq C(\bar\mu)
\qquad\text{whenever }\bar\mu\leq\bar\mu'.
\end{equation}

Our primitives and our monotonicity condition differ. We begin with welfare functions and derive possibly set-valued choices from optimization. Welfare maximization satisfies the natural correspondence version of contraction consistency: an optimizer under a larger box remains optimal under a smaller box whenever it remains feasible. Exchangeability does not, however, imply persistence of every point-valued selection, even if the selection is continuous and consistent. Appendix~\ref{app:examples-ancillary} provides a two-coordinate counterexample in which the failure arises from indifference among welfare-maximizing choices, and shows that when the welfare-generated choice correspondence is point-valued at every availability, its unique choice function is persistent. What exchangeability implies without point-valuedness is the weaker and selection-robust statement

\begin{equation}\label{eq:selection-robust}
\mathcal R^G(\alpha,\bar\mu)
\leq_Q
\mathcal R^G(\alpha,\bar\mu'),
\end{equation}
which is sufficient for the selection used in the next proposal phase.

The difference appears in the update rule. Persistence permits Alkan and Gale to reduce an availability coordinate directly to the amount retained when rejection occurs. Under Q-monotonicity, our algorithm makes the more gradual update in (\ref{eq:da-update}), subtracting the mass rejected in the current round. We do not claim a speed comparison: the two procedures operate under different primitives and different assumptions. Our contribution is to show that the weaker monotonicity derived from submodular welfare still supports a deferred-acceptance existence argument, including when choices are set-valued.

\subsubsection{Galichon, Kominers, and Weber}

\citeasnoun{GalichonKominersWeber2019} develop a broad framework for matching with imperfectly transferable utility and unobserved heterogeneity. Their framework describes the feasible utility frontier of a match through a transfer function and includes fully transferable utility, classical NTU, collective models, and settings with taxes on transfers.

The burn-only frontier in this paper is the NTU specialization

\begin{equation}\label{eq:ntu-frontier}
\Psi_{xy}(U_{xy},V_{xy})
=\max\{U_{xy}-\alpha_{xy},V_{xy}-\gamma_{xy}\}=0.
\end{equation}

It permits costly concessions but no transfer of the utility being relinquished. Thus our equilibrium concept should not be read as encompassing the entire imperfect-transferability framework of Galichon, Kominers, and Weber. Rather, it isolates the burn-only boundary case and combines it with general welfare-generated, possibly set-valued aggregate choices.

The two approaches also emphasize different results. Galichon, Kominers, and Weber incorporate additive unobserved heterogeneity, establish existence and uniqueness under their assumptions, and develop equilibrium-computation and estimation methods. We use submodular welfare and the Q-order to construct an aggregate deferred-acceptance process; our theorem establishes existence through subsequential convergence but does not claim uniqueness. Random utility enters both frameworks, but here its primary role is to generate submodular welfare functions to which the rejection comparative statics of Section~\ref{sec:comparative-statics} apply.

\subsubsection{Galichon, Hsieh, and Jacquet}

\citeasnoun{GalichonHsiehJacquet2026} study aggregate stable matching with money burning from a complementary point of view. Their analysis is restricted to a smooth specification in which the relevant choices are point-valued. That structure delivers existence and uniqueness of equilibrium and convergence of their generalized deferred-acceptance procedure to the unique equilibrium. An earlier aggregate model of matching without transfers, in which rationing operates through waiting times, is developed by \citeasnoun{GalichonHsieh2017}.

The present paper instead begins with convex duality and the Q-order. It permits nonsmooth welfare functions and set-valued choice correspondences, and it uses Q-monotone rejection to obtain well-definedness and equilibrium existence through subsequential convergence. These weaker hypotheses do not yield the uniqueness or full-sequence convergence available in the smooth, point-valued model of Galichon, Hsieh, and Jacquet. The two papers therefore emphasize different margins: their paper exploits smoothness and single-valuedness for sharper equilibrium and convergence conclusions, whereas ours isolates the order-theoretic rejection property that survives nonsmoothness and multiplicity.

These comparisons delimit the contribution of the present section. Relative to Gale and Shapley, it replaces individual contracts by divisible masses and accommodates nonunique choices. Relative to Alkan and Gale, it derives a weaker rejection monotonicity from welfare primitives rather than imposing persistence on a point-valued choice oracle. Relative to Galichon, Hsieh, and Jacquet, it relaxes smoothness and single-valuedness at the cost of weaker convergence and uniqueness conclusions. Relative to Galichon, Kominers, and Weber, it focuses on the no-transfer burn-only frontier and supplies a deferred-acceptance construction based on Q-monotone rejection.\ifdefined\GHSreviewlines\linelabel{line:sec:conclusion}\fi



\section{Discussion and conclusion}\label{sec:conclusion}

This paper began with a simple feature of deferred acceptance: rejection must be permanent. In a finite one-to-one market, permanence is implemented contract by contract. In an aggregate market, however, proposals and rejections are vectors of masses, quantities may be divisible, and welfare maximization may produce several optimal choices. The ordinary coordinatewise order is too strong to describe the resulting adjustment. Our purpose has been to identify the weaker order that is sufficient for the rejection process and to derive that order from economically interpretable welfare primitives.

The theoretical contribution is a dual monotone-comparative-statics framework for quantity choices. The P-order compares functions through the meet and join operations that underlie Topkis's theorem. The Q-order replaces those lattice operations with compensated exchanges between quantity vectors. For singleton-valued choices, the Q-order reduces to the ordinary coordinatewise order. For correspondences, it permits incomparable selections while ensuring that an excess in one selection can be transferred, with an admissible compensation in other coordinates, so as to preserve optimality.

The central duality result gives this order an accessible primitive characterization. A proper closed convex function is exchangeable if and only if its Legendre--Fenchel conjugate is submodular. Consequently, familiar submodularity restrictions on an indirect welfare or profit function imply Q-monotonicity of the quantity choices generated by its conjugate. The converse comparative-statics implication holds when the conjugate has a sublattice domain. This conclusion remains meaningful without differentiability or uniqueness. In the constrained production interpretation, it governs both responses that motivated the paper: how rejected quantities change when relative prices change and how they change when capacities expand.

The matching application shows why this distinction matters. In the deferred-acceptance welfare algorithm, one side proposes a welfare-maximizing vector of masses, the other side tentatively retains a welfare-maximizing subvector, and the rejected mass is removed from the available opportunities. The Q-monotonicity result ensures that this update can be iterated even when choices are set-valued. The availability sequence and the least waiting-time multiplier sequences converge monotonically; compactness supplies a convergent proposal subsequence; and every resulting joint limit constitutes an aggregate equilibrium after normalization of saturated coordinates. Thus the comparative-statics theory supplies both the economic logic and the technical device behind the equilibrium-existence argument.

The same formulation covers several environments that are usually treated separately. When each type contains a single agent and preferences are deterministic and tie-free, welfare maximization selects integral vertices, and the proposal and rejection logic specializes to classical deferred acceptance. With non-atomic type masses, it works directly with aggregate quantities rather than artificially labeling replicas. Welfare may be deterministic or generated by random utility, and the logit example yields an explicit aggregate matching system. These cases are unified not by assuming that all choices are point-valued, but by placing a common exchange restriction on the welfare functions that generate them.


\bibliographystyle{ECMA}
\bibliography{refAll}

\clearpage
\section*{Formal-verification note}\label{sec:formal-verification}

All formal mathematical results stated and proved in the revision have
machine-checked counterparts in Lean~4. The formalized development covers the Q-order and its set-valued
consequences; the exchangeability--submodularity conjugacy; the forward price
and capacity comparative statics used by the algorithm; the persistence results
for welfare-generated choices; well-definedness of the
aggregate deferred-acceptance iteration; compactified dual attainment and
the associated exact Kuhn--Tucker characterization; convergence of the
canonical least multipliers; and the corrected exact-equilibrium result.
On the matching side, the formalization uses lower semicontinuity and
convexity of the conjugates $G^*$ and $H^*$. These properties are automatic
for convex conjugates, while compactness of their domains is imposed in
Definition~\ref{def:welfare-function}. The development uses Mathlib only,
contains no \texttt{sorry}, invokes no nonstandard axioms, and builds against a
pinned Mathlib commit.
The complete Lean~4 development, pinned toolchain, and build instructions are
available at \url{https://github.com/alfredgalichon/lean-ghs}.

The converse in Theorem~\ref{thm:rejection-comparative-statics} is formalized
under the stated assumption that $\operatorname{dom}c^*$ is a sublattice. The
formal proof works on the effective domain of the extended-valued conjugate,
uses relative subgradients on its relative interior, and verifies the link with
the argmin rejection correspondence. A companion verification supplement gives
a result-by-result correspondence between the paper and the Lean source files
and provides the pinned build instructions. The two appendices that follow give
the mathematical proofs used in the paper.

\appendix

\section{Proofs of the comparative-statics results}\label{app:comparative-statics}

This appendix proves the results of Section~\ref{sec:comparative-statics}. All spaces are finite dimensional. We write $x^+=x\vee0$, $x^-=(-x)\vee0$, and $[a,b]=\{d:a\leq d\leq b\}$. A function is \emph{closed} when it is lower semicontinuous, and $\Gamma_0(\mathbb R^{\mathcal Z})$ denotes the proper closed convex functions from $\mathbb R^{\mathcal Z}$ to $\mathbb R\cup\{+\infty\}$.

Two qualifications used below are also imposed in the corresponding statements in Section~\ref{sec:comparative-statics}. First, the weak-set-order consequence of the P-order or Q-order requires both sets to be nonempty. Second, the implication from subgradient or rejection comparative statics back to submodularity assumes that $\operatorname{dom}c^*$ is a sublattice. Neither qualification is needed for the central equivalence between exchangeability of $c$ and submodularity of $c^*$.

\subsection{Order lemmas}

We begin with the common argument behind Topkis's theorem and its Q-order counterpart: ordered objective functions have ordered minimizers.

\begin{lemma}[minimizers under the P- and Q-orders]\label{lem:ordered-minimizers}
Let $f$ and $g$ be proper extended-real-valued functions.

(i) If $f\leq_P g$, then

\begin{equation}\label{eq:A-topkis}
\operatorname{argmin}f\leq_P\operatorname{argmin}g.
\end{equation}

(ii) If $f\leq_Qg$, then

\begin{equation}\label{eq:A-dual-topkis}
\operatorname{argmin}f\leq_Q\operatorname{argmin}g.
\end{equation}

Both conclusions remain valid for constrained minima after adding the relevant indicator functions to the objectives, provided the resulting constrained objectives are proper.
\end{lemma}

\begin{proof}
If either minimizer set is empty, the conclusion is immediate. Otherwise, take $x\in\operatorname{argmin}f$ and $y\in\operatorname{argmin}g$. Because $f$ and $g$ are proper, $f(x)$ and $g(y)$ are finite. In the P-order case,

\[
f(x\wedge y)+g(x\vee y)\leq f(x)+g(y).
\]

Minimality gives the reverse inequality term by term. Equality therefore holds, $x\wedge y$ minimizes $f$, and $x\vee y$ minimizes $g$.

In the Q-order case, fix $\delta_1\in[0,(x-y)^+]$. There is a $\delta_2\in[0,(x-y)^-]$ such that, with

\[
\widetilde x=x-\delta_1+\delta_2,
\qquad
\widetilde y=y+\delta_1-\delta_2,
\]
one has $f(\widetilde x)+g(\widetilde y)\leq f(x)+g(y)$. Minimality again gives the reverse inequality term by term. Hence $\widetilde x$ and $\widetilde y$ are minimizers. The constrained statements follow by applying the same arguments to $f+\iota_X$ and $g+\iota_Y$.
\end{proof}

The next elementary fact will be useful when we pass from Q-orders to coordinatewise comparisons.

\begin{lemma}[P-order, Q-order, and weak set order]\label{lem:orders-imply-weak-order}
Let $X,Y\subseteq\mathbb R^{\mathcal Z}$ be nonempty. If $X\leq_PY$ or $X\leq_QY$, then $X$ is below $Y$ in the weak set order: for every $x\in X$ there is a $y\in Y$ with $x\leq y$, and for every $y\in Y$ there is an $x\in X$ with $x\leq y$.
\end{lemma}

\begin{proof}
Suppose first that $X\leq_PY$. Fix $x\in X$ and $y\in Y$. Then $x\wedge y\in X$, $x\vee y\in Y$, and

\[
x\wedge y\leq y,
\qquad
x\leq x\vee y.
\]

Now suppose that $X\leq_QY$. Fix any $x\in X$ and $y\in Y$, and apply the Q-order with the full positive discrepancy $\delta_1=(x-y)^+$. It supplies $\delta_2\in[0,(x-y)^-]$ and points

\[
\widetilde x=x-(x-y)^++\delta_2\in X,
\qquad
\widetilde y=y+(x-y)^+-\delta_2\in Y.
\]

On coordinates where $x>y$, one has $\widetilde y=x$; on the remaining coordinates, the bound on $\delta_2$ gives $\widetilde y\geq x$. Thus $x\leq\widetilde y$. The same coordinatewise calculation gives $\widetilde x\leq y$. A single exchange therefore supplies both dominating points.
\end{proof}

Nonemptiness in Lemma~\ref{lem:orders-imply-weak-order} is essential. If $X\neq\varnothing$ and $Y=\varnothing$, both the P-order and the Q-order hold vacuously whereas the weak set order does not.

\subsection{Duality of the P- and Q-orders}

We now prove Theorem~\ref{thm:pq-duality}. The forward implication is elementary and does not use convexity. The reverse implication uses Fenchel--Moreau duality and a minimax argument.

\begin{theorem}[P--Q conjugacy]\label{thm:appendix-pq-conjugacy}
Let $f,g\in\Gamma_0(\mathbb R^{\mathcal Z})$. Then

\begin{equation}\label{eq:A-pq-conjugacy}
f\leq_Qg
\quad\Longleftrightarrow\quad
g^*\leq_Pf^*.
\end{equation}
\end{theorem}
\begin{proof}
\emph{Proof of the forward implication.} Assume $f\leq_Qg$. Fix $p,p'\in\mathbb R^{\mathcal Z}$, $q\in\operatorname{dom}f$, and $q'\in\operatorname{dom}g$. Define

\begin{equation}\label{eq:A-delta1}
(\delta_1)_i=
\begin{cases}
    (q_i-q_i')^+ & \text{if } p_i> p'_i \\
    0 & \text{otherwise}.
\end{cases}
\end{equation}

Choose the compensation $\delta_2\in[0,(q-q')^-]$ supplied by the Q-order, and set $\widetilde q=q-\delta_1+\delta_2$ and $\widetilde q'=q'+\delta_1-\delta_2$. A coordinatewise calculation gives

\begin{equation}\label{eq:A-coordinatewise}
q^\top(p\vee p')+(q')^\top(p\wedge p')
\leq
(\widetilde q)^\top p'+(\widetilde q')^\top p.
\end{equation}

Indeed, if $p_i>p_i'$, the difference between the right- and left-hand sides in coordinate $i$ is

\[
(p_i-p_i')(q_i'-q_i+(\delta_1)_i-(\delta_2)_i)\geq0;
\]
if $p_i\leq p_i'$, then $(\delta_1)_i=0$ and the difference is $(p_i-p_i')(-(\delta_{2})_i)\geq0$. Combining (\ref{eq:A-coordinatewise}) with

\[
f(\widetilde q)+g(\widetilde q')\leq f(q)+g(q')
\]
yields

\begin{equation}\label{eq:A-forward-chain}
\begin{aligned}
q^\top(p\vee p')-f(q)
 +(q')^\top(p\wedge p')-g(q')
&\leq (\widetilde q)^\top p'-f(\widetilde q)
 +(\widetilde q')^\top p-g(\widetilde q')\\
&\leq f^*(p')+g^*(p).
\end{aligned}
\end{equation}

Taking the supremum over $q$ and $q'$ proves

\[
g^*(p\wedge p')+f^*(p\vee p')\leq g^*(p)+f^*(p'),
\]
which is $g^*\leq_Pf^*$.

For the \emph{reverse implication}, it suffices to consider
$x\in\dom f$ and $y\in\dom g$, since otherwise the defining Q-order
inequality is automatic. Fix such $x,y$ and
$\delta_1\in[0,(x-y)^+]$. Let

\[
K=[0,(x-y)^-],
\qquad
D=\operatorname{dom}f^*\times\operatorname{dom}g^*.
\]

The set $K$ is compact and convex, while $D$ is convex. On $K\times D$, define

\[
L(\delta_2,\alpha,\beta)
=\alpha^\top(x-\delta_1+\delta_2)
+\beta^\top(y+\delta_1-\delta_2)
-f^*(\alpha)-g^*(\beta).
\]

For fixed $(\alpha,\beta)$, the function $L$ is continuous and affine in $\delta_2$. For fixed $\delta_2$, it is concave and upper semicontinuous in $(\alpha,\beta)$ because $f^*$ and $g^*$ are lower semicontinuous. Thus the minimax theorem of \citeasnoun{Sion1958} applies and gives

\begin{equation}\label{eq:A-sion}
\min_{\delta_2\in K}\sup_{(\alpha,\beta)\in D}L(\delta_2,\alpha,\beta)
=\sup_{(\alpha,\beta)\in D}\min_{\delta_2\in K}L(\delta_2,\alpha,\beta).
\end{equation}

By Fenchel--Moreau duality, the left-hand side is

\[
\min_{\delta_2\in K}
\{f(x-\delta_1+\delta_2)+g(y+\delta_1-\delta_2)\}.
\]

The minimum is attained because the objective is lower semicontinuous on the compact set $K$; if it is identically $+\infty$, every point of $K$ is a minimizer. The right-hand side of (\ref{eq:A-sion}) is

\[
\sup_{(\alpha,\beta)\in D}
\{\alpha^\top(x-\delta_1)+\beta^\top(y+\delta_1)
-\langle(\alpha-\beta)^-,(x-y)^-\rangle
-f^*(\alpha)-g^*(\beta)\}.
\]

Taking the supremum over $\delta_1\in[0,(x-y)^+]$ and interchanging the two suprema gives
    \begin{equation}\label{eq:A-dual-value}
        \sup_{(\alpha,\beta)\in D} (\alpha\vee \beta)^\top x+(\alpha \wedge \beta)^\top y - f^\ast(\alpha) - g^\ast(\beta)
    \end{equation}
    is equal to
    \begin{equation}\label{eq:A-primal-value}
        \sup_{\delta_1 \in [0,(x-y)^+] } \min_{\delta_2 \in [0,(x-y)^-]} f(x-\delta_1+\delta_2) + g(y+\delta_1-\delta_2).
    \end{equation}
Assume now that $g^*\leq_P f^*$. For every $\alpha,\beta$,

\[
g^*(\alpha\wedge\beta)+f^*(\alpha\vee\beta)
\leq g^*(\beta)+f^*(\alpha).
\]

Consequently, the quantity in (\ref{eq:A-dual-value}) is at most $f(x)+g(y)$. Its equality with (\ref{eq:A-primal-value}) implies that, for every $\delta_1\in[0,(x-y)^+]$, the attained minimum over $\delta_2\in[0,(x-y)^-]$ is at most $f(x)+g(y)$. This is exactly $f\leq_Q g$ and proves the reverse implication.
\end{proof}
\begin{corollary}[support functions]\label{cor:support-functions}
If $X,Y\subseteq\mathbb R^{\mathcal Z}$ are nonempty closed convex sets and $\sigma_X,\sigma_Y$ are their support functions, then

\begin{equation}\label{eq:A-support}
X\leq_QY
\quad\Longleftrightarrow\quad
\sigma_Y\leq_P\sigma_X.
\end{equation}
\end{corollary}

\begin{proof}
Apply Theorem~\ref{thm:appendix-pq-conjugacy} to $\iota_X$ and $\iota_Y$, and use $\iota_X^*=\sigma_X$ and $\iota_Y^*=\sigma_Y$.
\end{proof}

Setting $f=g=c$ in Theorem~\ref{thm:appendix-pq-conjugacy} proves Theorem~\ref{thm:exchangeability-submodularity}:

\begin{equation}\label{eq:A-exch-submod}
c\text{ is exchangeable}
\quad\Longleftrightarrow\quad
c^*\text{ is submodular}.
\end{equation}

Together with Lemma~\ref{lem:ordered-minimizers}(ii), Theorem~\ref{thm:appendix-pq-conjugacy} also proves Theorem~\ref{thm:dual-topkis} and shows that every nonempty minimizer set of an exchangeable objective is a matron.

\subsection{Exchangeability implies convexity}

Section~\ref{sec:comparative-statics} states that lower semicontinuous exchangeable functions are convex. We give the argument because it also justifies the convexity of closed matrons.

\begin{proposition}\label{prop:exchangeability-implies-convexity}
Let $f:\mathbb R^{\mathcal Z}\to\mathbb R\cup\{+\infty\}$ be proper, lower semicontinuous, and exchangeable. Then $f$ is convex. Consequently, every nonempty closed matron is convex.
\end{proposition}

\begin{proof}
The lower semicontinuity of $f$ ensures that it is enough to prove midpoint convexity. Fix $q,q'\in\operatorname{dom}f$ and put $m=(q+q')/2$. Consider the pairs $(u,u')$ satisfying

\begin{equation}\label{eq:A-pairs}
u+u'=q+q',\qquad
u,u'\in[q\wedge q',q\vee q'],\qquad
f(u)+f(u')\leq f(q)+f(q').
\end{equation}

This set is nonempty and compact (by being closed and contained in a compact box). Choose a pair minimizing the continuous function
\begin{equation}\label{eq:A-l1}
(u,u') \mapsto \lVert u-m\rVert_1+\lVert u'-m\rVert_1.
\end{equation}
We reason by contradiction and assume that $u \neq u'$. Set $d=u-u'$. After interchanging $u$ and $u'$ if necessary, we may assume that $d^+\neq0$. Apply exchangeability with $\delta_1=d^+/2$. There is $\delta_2 \in [0,d^-]$ such that
\begin{equation*}
    f(u-\delta) + f(u'+\delta) \leq f(u)+f(u'),
\end{equation*}
where $\delta = \delta_1-\delta_2$.
Thus $(u-\delta,u'+\delta)$ satisfies (\ref{eq:A-pairs}). For every coordinate with $d_z>0$, the new difference between the two vectors is zero. For every coordinate with $d_z\leq0$, it is $d_z+2(\delta_2)_z$, whose absolute value is at most $|d_z|$ because $0\leq(\delta_2)_z\leq-d_z$. Hence the distance in (\ref{eq:A-l1}) weakly decreases in every coordinate and strictly decreases in at least one coordinate:
\begin{equation*}
    \lVert u-\delta-m\rVert_1+\lVert u'+\delta-m\rVert_1 < \lVert u-m\rVert_1+\lVert u'-m\rVert_1
\end{equation*}
which contradicts the optimality of $(u,u')$.
Therefore $u=u'=m$, and (\ref{eq:A-pairs}) yields midpoint convexity

\begin{equation}\label{eq:A-midpoint}
2f\left(\frac{q+q'}2\right)\leq f(q)+f(q').
\end{equation}
\end{proof}

\subsection{Rejected quantities}

We turn to Theorem~\ref{thm:rejection-comparative-statics}. It is convenient to write the rejection correspondence directly as an argmin. For

\[
\mathcal Q(p,\bar q)
=\operatorname*{argmax}_{q\leq\bar q}
\{p^\top q-c(q)\},
\]
the change of variables $r=\bar q-q$ gives

\begin{equation}\label{eq:A-rejection-argmin}
\mathcal R(p,\bar q)
=\operatorname*{argmin}_{r\geq0}
\{c(\bar q-r)+p^\top r\}.
\end{equation}

\subsubsection{From exchangeability to rejection monotonicity}
\begin{lemma}[exchangeability implies rejection monotonicity]\label{lem:exchangeability-rejection}
    Let $c\in\Gamma_0(\mathbb R^\mathcal Z)$ be exchangeable. Then for
    $p\leq p'$ and $\bar q \leq \bar q'$ such that there exists some $q\leq\bar q$ in the domain of $c$, we have
    \begin{equation*}
    \Pi_D \mathcal{R}(p,\bar q) \leq_Q \Pi_D \mathcal{R}(p',\bar{q}'),
\end{equation*}
where $D = \{i \mid p_i = p'_i\}$.
\end{lemma}

\begin{proof}
Let $p\leq p'$ and $\bar q \leq \bar q'$.
If either rejection correspondence is empty, the conclusion is immediate. Otherwise, the feasibility assumption makes both constrained objectives proper, so their values at their minimizers are finite.
Let $r_D \in \Pi_D \mathcal{R}(p,\bar q)$ and $r'_D \in \Pi_D \mathcal{R}(p',\bar{q}')$. Let $r \in \mathcal{R}(p,\bar q), r' \in \mathcal{R}(p',\bar q') $ be such that $\Pi_D(r)= r_D, \Pi_D(r') = r'_D$. Let $\delta_{1D} \in [0,(r_D-r'_D)^+]$ and set
\begin{equation*}
    (\delta_1)_i =
    \begin{cases}
        (\delta_{1D})_i & i \in D\\
        0 & i \notin D
    \end{cases}.
\end{equation*}
By construction $\delta_1 \in [0,((\bar q' - r')-(\bar q-r))^+]$ because $\bar q \leq \bar q'$ implies $(r-r')^+ \leq ((\bar q' - r')-(\bar q-r))^+$. The exchangeability of $c$ gives $\delta_2 \in [0,((\bar q' - r')-(\bar q-r))^-]$ such that
\begin{equation*}
    c(\bar q - (r - \delta_1 + \delta_2)) + c(\bar q' - (r'+\delta_1-\delta_2)) \leq c(\bar q - r) + c(\bar q' -r').
\end{equation*}
The exchanged rejection vectors are feasible. Indeed,
\[
0\leq\delta_1\leq(r-r')^+\leq r,
\qquad
0\leq\delta_2
\leq((\bar q'-\bar q)+r-r')^-
\leq(r-r')^-\leq r',
\]
so $r-\delta_1+\delta_2\geq0$ and $r'+\delta_1-\delta_2\geq0$.
Moreover
\begin{equation*}
    p^\top(r - \delta_1 + \delta_2) + p'^\top(r'+\delta_1-\delta_2) \leq p^\top r + p'^\top r',
\end{equation*}
because $p\leq p'$ implies $(p-p')^\top \delta_2 \leq 0$ and $(\delta_1)_i = 0$ for $ i \notin D$. Summing the two objective inequalities, and using the optimality of $r$ and $r'$, shows that both feasible exchanged vectors are minimizers:
\begin{equation*}
    r-\delta_1 + \delta_2 \in \mathcal{R}(p,\bar q), \quad r'+\delta_1 -\delta_2 \in \mathcal{R}(p',\bar q').
\end{equation*}
Projecting onto $D$, the preceding bound gives $\delta_{2D}=\Pi_D(\delta_2)\in[0,(r_D-r'_D)^-]$. Hence
\begin{equation*}
    r_D-\delta_{1D}+\delta_{2D} \in \Pi_D\mathcal{R}(p,\bar q), \quad  r'_D+\delta_{1D}-\delta_{2D} \in \Pi_D\mathcal{R}(p',\bar q').
\end{equation*}
\end{proof}

\subsubsection{From rejection monotonicity to subgradient monotonicity}

 We use the Fenchel--Young equivalence

\begin{equation}\label{eq:A-fenchel-young}
q\in\partial c^\ast(p)
\quad\Longleftrightarrow\quad
q\in\operatorname*{argmax}_{z\in\mathbb R^{\mathcal Z}}
\{p^\top z-c(z)\}.
\end{equation}

\begin{lemma}[rejection monotonicity implies subgradient monotonicity]\label{lem:rejection-subgradient}
Suppose that item 3 of Theorem~\ref{thm:rejection-comparative-statics} holds. For every $p,p'\in\operatorname{relint}(\operatorname{dom}c^*)$ with $p\leq p'$, there are $q\in\partial c^*(p)$ and $q'\in\partial c^*(p')$ such that

\[
q'_i\leq q_i
\qquad\text{whenever }p_i=p'_i.
\]
\end{lemma}

\begin{proof}
The subdifferentials of $c^*$ are nonempty on the relative interior of its domain. Choose $q^0\in\partial c^*(p)$ and $q^{0\prime}\in\partial c^*(p')$, and put $\bar q=q^0\vee q^{0\prime}$. Both vectors are feasible for their respective capacity-constrained problems and attain the unconstrained maximum. Hence

\[
\varnothing\neq\mathcal Q(p,\bar q)\subseteq\partial c^*(p),
\qquad
\varnothing\neq\mathcal Q(p',\bar q)\subseteq\partial c^*(p').
\]

Let $D=I(p,p')$. The rejection response with the same capacity on both sides gives

\[
\Pi_D\mathcal R(p,\bar q)
\leq_Q
\Pi_D\mathcal R(p',\bar q).
\]

By Lemma~\ref{lem:orders-imply-weak-order}, there are $r\in\mathcal R(p,\bar q)$ and $r'\in\mathcal R(p',\bar q)$ such that $r_D\leq r'_D$. Set $q=\bar q-r$ and $q'=\bar q-r'$. The preceding inclusion gives $q\in\partial c^*(p)$ and $q'\in\partial c^*(p')$, while $r_D\leq r'_D$ is equivalent to $q'_D\leq q_D$.
\end{proof}

\begin{lemma}[subgradient monotonicity implies submodularity]\label{lem:subgradient-submodularity}
Let $h:\mathbb R^{\mathcal Z}\to\mathbb R\cup\{+\infty\}$ be proper, convex, and lower semicontinuous. Suppose that $\operatorname{dom}h$ is a sublattice and that for every $p,p' \in \operatorname{relint}(\operatorname{dom}h)$ satisfying $p\leq p'$, there are $q\in\partial h(p)$ and $q'\in\partial h(p')$ satisfying

\[
q'_i\leq q_i
\qquad\text{for every }i\text{ such that }p_i=p_i'.
\]

Then $h$ is submodular.
\end{lemma}
\begin{proof}
First take $p,p'\in\operatorname{relint}(\operatorname{dom}h)$. Since the domain is a sublattice, $p\wedge p'$ and $p\vee p'$ belong to $\operatorname{dom}h$. Put $d=p-p\wedge p'\geq0$ and define

\[
\phi(t)=h(p\wedge p'+td),
\qquad
\psi(t)=h(p'+td),
\qquad 0\leq t\leq1.
\]

For $0<t<1$,

\[
p\wedge p'+td=(1-t)(p\wedge p')+tp,
\qquad
p'+td=(1-t)p'+t(p\vee p').
\]

The open segment joining a point of a convex set to a point of its relative interior lies in the relative interior. Thus both displayed points belong to $\operatorname{relint}(\operatorname{dom}h)$. They are coordinatewise ordered, and their coordinates agree whenever $d_i>0$. The assumed subgradient property supplies $q$ and $q'$ at these two points with $q'_i\leq q_i$ whenever $d_i>0$. Therefore
\[
\psi'_-(t)\leq\langle q',d\rangle
\leq\langle q,d\rangle\leq\phi'_+(t),
\]
where the outer inequalities follow from \cite[Theorem 23.2]{rockafellar-1970a}. At every point where both one-dimensional convex functions are differentiable, this gives $\psi'(t)\leq\phi'(t)$. Integrating first on $[\varepsilon,1-\varepsilon]$ and then letting $\varepsilon\downarrow0$, using lower semicontinuity and convexity for the one-sided endpoint limits, yields

\[
h(p\vee p')-h(p')\leq h(p)-h(p\wedge p').
\]

It remains to extend the inequality to the boundary of the domain. Choose $a\in\operatorname{relint}(\operatorname{dom}h)$. For $p,p'\in\operatorname{dom}h$ and $0<\varepsilon<1$, set

\[
p_\varepsilon=(1-\varepsilon)p+\varepsilon a,
\qquad
p'_\varepsilon=(1-\varepsilon)p'+\varepsilon a.
\]

These points lie in the relative interior, and coordinatewise distributivity gives

\[
p_\varepsilon\wedge p'_\varepsilon
=(1-\varepsilon)(p\wedge p')+\varepsilon a,
\qquad
p_\varepsilon\vee p'_\varepsilon
=(1-\varepsilon)(p\vee p')+\varepsilon a.
\]

Apply the preceding inequality to $p_\varepsilon,p'_\varepsilon$. For any $x\in\operatorname{dom}h$, convexity gives
\[
\limsup_{\varepsilon\downarrow0}h((1-\varepsilon)x+\varepsilon a)
\leq h(x),
\]
while lower semicontinuity gives the reverse liminf inequality. Hence all four terms converge to their boundary values, and passage to the limit proves the submodular inequality for $p,p'\in\operatorname{dom}h$. If either point lies outside the domain, the inequality is automatic in the extended-real sense. Thus $h$ is submodular.
\end{proof}

\subsubsection{Completion of the proof of Theorem~\ref{thm:rejection-comparative-statics}}

We can now assemble the equivalences.

\begin{enumerate}

\item[1.] Exchangeability of $c$ is equivalent to submodularity of $c^*$ by Theorem~\ref{thm:appendix-pq-conjugacy} with $f=g=c$.

\item[2.] Exchangeability implies the joint price-and-capacity rejection response by Lemma~\ref{lem:exchangeability-rejection}.

\item[3.] The rejection response implies the subgradient monotonicity condition in Lemma~\ref{lem:rejection-subgradient}.

\item[4.] Under the sublattice assumption on $\operatorname{dom}c^*$, that subgradient condition implies submodularity of $c^*$ by Lemma~\ref{lem:subgradient-submodularity}.

\end{enumerate}

This proves the full equivalence stated in Theorem~\ref{thm:rejection-comparative-statics}.

\subsection{Examples and ancillary claims}\label{app:examples-ancillary}

\subsubsection{Separable costs and entropy}

Let $c(q)=\sum_i c_i(q_i)$ with each $c_i$ closed, proper, and convex. For $q,q'$ in the domain and $\delta_1\in[0,(q-q')^+]$, choose $\delta_2=0$. In each affected coordinate, the two arguments move toward one another while preserving their sum. Scalar convexity therefore gives

\[
c(q-\delta_1)+c(q'+\delta_1)
\leq c(q)+c(q'),
\]
so $c$ is exchangeable. Equivalently, $c^*(p)=\sum_i c_i^*(p_i)$ is modular.

For entropy on the simplex,

\[
c(\mu)=\sum_i\mu_i\log\mu_i+\iota_\Delta(\mu),
\qquad
\Delta=\left\{\mu\geq0:\sum_i\mu_i=1\right\},
\]
the conjugate is

\[
c^*(p)=\log\sum_i e^{p_i}.
\]

Its off-diagonal Hessian entries are $-s_i(p)s_j(p)\leq0$, where $s_i(p) = \frac{e^{p_i}}{\sum_j e^{p_j}}$. Hence $c^*$ is submodular and $c$ is exchangeable by Theorem~\ref{thm:appendix-pq-conjugacy}.

\subsubsection{Quadratic costs}

Let $M$ be symmetric positive definite and

\[
c(q)=\frac12q^\top Mq.
\]

Completing the square gives

\[
c^*(p)=\frac12p^\top M^{-1}p.
\]

A continuously differentiable quadratic function is submodular exactly when every off-diagonal entry of its Hessian is nonpositive. Theorem~\ref{thm:appendix-pq-conjugacy} therefore gives

\begin{equation}\label{eq:A-quadratic}
c\text{ is exchangeable}
\quad\Longleftrightarrow\quad
(M^{-1})_{ij}\leq0\quad(i\neq j).
\end{equation}

For the matrix in (\ref{eq:quadratic-example}),

\[
M^{-1}=\begin{pmatrix}1&-2\\-2&5\end{pmatrix},
\]
so (\ref{eq:A-quadratic}) proves exchangeability. The failure of diagonal dominance in the first row shows that the example need not satisfy the stronger quadratic condition for continuous $M^\natural$-convexity.

Under the capacity constraint $q\leq\bar q$, let $S=M^{-1}$ and let $\tau\geq0$ be the Kuhn--Tucker multiplier. Stationarity and complementarity are

\[
Mq-p+\tau=0,
\qquad
\tau^\top(\bar q-q)=0.
\]

With $r=\bar q-q$ and $\rho=\bar q-Sp$, this is equivalent to

\begin{equation}\label{eq:A-lcp}
r=S\tau+\rho,
\qquad
r\geq0,\quad\tau\geq0,\quad r^\top\tau=0,
\end{equation}
which proves (\ref{eq:lcp}).

\subsubsection{Relation with continuous discrete-convex classes}\label{app:continuous-dca}

Proposition~\ref{prop:m-natural-exchangeable} follows by combining the continuous conjugacy theorem of \citeasnoun[Theorem 1.1]{MurotaShioura2004} with Theorem~\ref{thm:exchangeability-submodularity}. If $c$ is continuous $M^\natural$-convex, its conjugate $c^*$ is $L^\natural$-convex; $L^\natural$-convexity implies ordinary submodularity; and Theorem~\ref{thm:exchangeability-submodularity} then implies exchangeability of $c$. The quadratic example above shows that the reverse implication fails.

\subsubsection{Relation with S-convexity}
\begin{proof}[Proof of Proposition \ref{prop:s-convex-exchangeable}]
    \citeasnoun[Theorem 2]{chen2022s} show that if $c$ is lower-semicontinuous and S-convex, then its conjugate $c^*$ is convex and submodular. Theorem~\ref{thm:exchangeability-submodularity} then implies that $c$ is exchangeable, proving Proposition~\ref{prop:s-convex-exchangeable}. \citeasnoun[Theorem 3]{chen2022s} also show that the converse holds when $c$ is continuously differentiable and strictly convex with open effective domain.
\end{proof}

\subsubsection{Persistence of welfare-generated choices}\label{app:persistence}

\begin{proposition}\label{prop:persistence-choices}
\begin{enumerate}
\item[(i)] There is a submodular welfare function $G$ for which
$\mathcal Q^G(0,\cdot)$ admits a continuous and consistent point-valued
selection that is not persistent.
\item[(ii)] Let $G$ be a submodular welfare function and fix $\alpha$. If
$\mathcal Q^G(\alpha,\bar\mu)$ is a singleton for every $\bar\mu\geq0$, its
unique choice function is persistent.
\end{enumerate}
\end{proposition}

\begin{proof}
For (i), take two coordinates and let
\[
G^*(\mu)=\iota_{[0,1]^2}(\mu).
\]
\noindent Equivalently,
$G(U)=\max\{0,U_1\}+\max\{0,U_2\}$ is submodular. At $\alpha=0$,
\[
\mathcal Q^G(0,\bar\mu)=[0,\bar\mu\wedge\mathbf 1].
\]
\noindent From this correspondence select
\[
C(\bar\mu)=\operatorname*{argmax}_{0\leq\mu\leq\bar\mu\wedge\mathbf 1}
\left\{-(\mu_1-\mu_2)^2-(\mu_1-1)^2-(\mu_2-1)^2\right\}.
\]
\noindent The objective is strictly concave, so the selection is point-valued
and continuous. It is also consistent: if $\bar\mu\leq\bar\mu'$ and
$C(\bar\mu')\leq\bar\mu$, the unique maximizer over the larger box remains
feasible and is therefore the unique maximizer over the smaller box. However,
\[
C(1,1/2)=(3/4,1/2),
\qquad
C(1,1)=(1,1),
\]
\noindent and hence
\[
C(1,1)\wedge(1,1/2)=(1,1/2)
\nleq(3/4,1/2)=C(1,1/2).
\]
\noindent Thus persistence fails. The failure exploits indifference in the
original welfare problem and the freedom to select among its multiple optima.

For (ii), let $C(\bar\mu)$ denote the unique element of
$\mathcal Q^G(\alpha,\bar\mu)$. For $\bar\mu\leq\bar\mu'$, write
$q=C(\bar\mu)$, $q'=C(\bar\mu')$, and
$\widetilde\mu=\bar\mu\vee q'$. Contraction consistency gives
$C(\widetilde\mu)=q'$. The singleton form of the capacity-induced response
(\ref{eq:capacity-response-G}) gives
\[
\bar\mu-q\leq\widetilde\mu-q'.
\]
\noindent Coordinatewise, this inequality is equivalent to
$q'\wedge\bar\mu\leq q$, which is persistence.
\end{proof}

\par

\section{Proofs of the matching and convergence results}\label{app:matching}

This appendix proves Proposition~\ref{prop:classical-welfare-representation} and Theorem~\ref{thm:aggregate-da}. All index sets are
finite, $\E=\mathcal X\times\mathcal Y$, and vector inequalities are
coordinatewise. By Definition~\ref{def:welfare-function}, $G$ is proper, convex and lower
semicontinuous on $\R^\E$, and $\dom G^\ast\subseteq\R_+^\E$ is compact and
contains $0$. Hence
\begin{equation}
G(a)=\max_{\mu\in\dom G^\ast}\{\mu^\top a-G^\ast(\mu)\}\qquad(a\in\R^\E),
\label{eq:Grep}
\end{equation}
holds by the Fenchel--Moreau theorem, the maximum
being attained by upper semicontinuity on a compact set; $G$ is finite on
$\R^\E$ with $G\geq-G^\ast(0)$; and $G^\ast$ attains a finite minimum on
$\dom G^\ast$. The same holds for $H$.

\subsection{Classical stability and the welfare representation}

Recall
\begin{equation}
G_0(U)=\sum_x\max\{\alpha_{x0},\max_y U_{xy}\},
\qquad
H_0(V)=\sum_y\max\{\gamma_{0y},\max_x V_{xy}\}.
\label{eq:B-classical-welfare}
\end{equation}
For $\mu\in \R_+^{\mathcal X\times\mathcal Y}$, $\mu\in\partial G_0(U)$ is
equivalent to
\begin{equation}
\sum_y\mu_{xy}\leq 1,
\qquad
\mu_{xy}>0\Longrightarrow
U_{xy}=\max\{\alpha_{x0},\max_{y'}U_{xy'}\}.
\label{eq:B-classical-subgradient}
\end{equation}
\noindent If the first inequality is strict, the outside option attains the
maximum. This is the usual subdifferential formula for a finite maximum of
affine functions \cite{rockafellar-1970a}. The analogous characterization
holds for $\partial H_0(V)$.

\begin{proof}[Proof of Proposition~\ref{prop:classical-welfare-representation}]
Suppose first that $\mu$ is a stable NTU matching. Set
\begin{equation}
U_{xy}=\min\{u_x^\mu,\alpha_{xy}\},
\qquad
V_{xy}=\min\{v_y^\mu,\gamma_{xy}\}.
\label{eq:B-classical-UV}
\end{equation}
The no-blocking condition in (\ref{eq:ir-stability}) gives
\begin{equation}
\max\{U_{xy}-\alpha_{xy},V_{xy}-\gamma_{xy}\}=0
\quad\text{for every }x,y.
\label{eq:B-classical-frontier}
\end{equation}
If $x$ is matched with $y$, then $u_x^\mu=\alpha_{xy}$, so
$U_{xy}=u_x^\mu$, while $U_{xy'}\leq u_x^\mu$ for every $y'$. If $x$ is
unmatched, individual rationality and $u_x^\mu=\alpha_{x0}$ imply that the
outside option is maximizing. Hence $\mu\in\partial G_0(U)$. The symmetric
argument gives $\mu\in\partial H_0(V)$.

Conversely, suppose that $\mu$ has an equilibrium representation $(U,V)$. If
$x$ is matched with $y$, choice optimality and the frontier condition give
\[
u_x^\mu=\alpha_{xy}\geq U_{xy}\geq\alpha_{x0}.
\]
If $x$ is unmatched, $u_x^\mu=\alpha_{x0}$. Thus every $x$ is individually
rational, and the same argument applies to every $y$. Choice optimality also
implies $u_x^\mu\geq U_{xy}$ and $v_y^\mu\geq V_{xy}$ for every pair $x,y$.
Consequently,
\[
\max\{u_x^\mu-\alpha_{xy},v_y^\mu-\gamma_{xy}\}
\geq
\max\{U_{xy}-\alpha_{xy},V_{xy}-\gamma_{xy}\}
=0.
\]
Hence $\mu$ has no blocking pair and is stable.
\end{proof}

\subsection{Box choice and compactified multipliers}

We endow $\underline\R=[-\infty,+\infty)$ and $\overline\R_+=[0,+\infty]$
with their order topologies, and their finite Cartesian powers with the
corresponding product topologies.

For $\bar\mu\geq0$ recall that
\begin{equation*}
\mathcal Q^G(\alpha,\bar\mu)
=\operatorname*{argmax}_{\mu\leq\bar\mu}
\{\mu^\top\alpha-G^\ast(\mu)\}.
\end{equation*}
Because $0\in\dom G^\ast$ and $\dom G^\ast$ is compact, this set is nonempty; it is compact by lower semicontinuity of $G^\ast$. These capacity constraints induce waiting times that may be infinite (e.g.\ if
$\bar\mu_{xy}=0$). To handle them, we construct the canonical continuous extension of the finite-domain welfare function $G$ to
$\underline\R^\E=[-\infty,+\infty)^\E$. For $\mu\in\R_+^\E$ and
$p\in\underline\R^\E$ let
\[
\langle\mu,p\rangle_-=
\begin{cases}
\sum_{z:\,p_z>-\infty}\mu_zp_z &\text{if }\mu_z=0\text{ whenever }p_z=-\infty,\\
-\infty&\text{otherwise}.
\end{cases}
\]
 Set
\begin{equation}
\hat G(p)=\sup_{\mu\in\R_+^\E}\{\langle\mu,p\rangle_--G^\ast(\mu)\}.
\label{eq:B2}
\end{equation}
On $\R^\E$, we have $\hat G=G$ by the Fenchel--Moreau theorem, since $G$ is
proper, closed, and convex.

\begin{lemma}\label{app:lemma-lsc-scalar-prod}
The map $(\mu,p)\mapsto\langle\mu,p\rangle_-$ is upper semicontinuous on
$\R_+^\E\times\underline\R^\E$, nondecreasing in $p$, and has increasing
differences: if $\mu\leq\mu'$ and $p,p'\in\underline\R^\E$, then
\begin{equation}
\langle\mu',p\rangle_-+\langle\mu,p'\rangle_-
\leq
\langle\mu',p\vee p'\rangle_-+\langle\mu,p\wedge p'\rangle_-.
\label{eq:B3}
\end{equation}
For fixed $\mu$, $p\mapsto\langle\mu,p\rangle_-$ is continuous.
\end{lemma}

\begin{proof}
Introduce $\varphi:\R_+\times\underline\R\to\underline\R$ by
$\varphi(m,t)=mt$, with the conventions $0\times(-\infty)=0$ and
$m\times(-\infty)=-\infty$ for $m>0$. Then $\langle \mu,p\rangle_- = \sum_z \varphi(\mu_z,p_z)$. It is thus enough to show the continuity property for $\varphi$.
 The function $\varphi$ is continuous at points with
$t>-\infty$; at $(m,-\infty)$ with $m>0$ it equals $-\infty$ and
$\varphi(m',t')\to-\infty$ as $(m',t')\to(m,-\infty)$; at $(0,-\infty)$ it equals
$0$ and $\varphi(m',t')\leq0$ nearby. Hence $\varphi$ is upper semicontinuous,
and $\varphi(m,\cdot)$ is continuous for fixed $m$. Monotonicity in $p$ is
immediate since $m\geq 0$. For (\ref{eq:B3}), we may assume that the left-hand side is finite (otherwise there is nothing to prove), so that
$\mu'_z=0$ where $p_z=-\infty$ and $\mu_z=0$ where $p'_z=-\infty$. Since
$\mu\leq\mu'$, we have that $\mu'_z=0$ whenever $(p\vee p')_z=-\infty$ and $\mu_z=0$ whenever $(p\wedge p')_z=-\infty$. At the remaining coordinates, both
sides coincide if $p_z\geq p'_z$, and otherwise the difference is
$(\mu'_z-\mu_z)(p'_z-p_z)\geq0$.
\end{proof}

\begin{proposition}\label{app:prop-hatG}
The function $\hat G:\underline\R^\E\to\R$ is well defined, nondecreasing,
continuous and submodular, and $\hat G\geq-G^\ast(0)$. The supremum in
(\ref{eq:B2}) is attained, and every maximizer $\mu$ satisfies $\mu_z=0$
whenever $p_z=-\infty$. In particular $\hat G$ is the unique continuous
extension of $G$: $\hat G(p)=\lim_{p_n\to p,\,p_n\in\R^\E}G(p_n)$.
\end{proposition}

\begin{proof}
Fix $p$ and $S=\{z\mid p_z=-\infty\}$. The bracket in (\ref{eq:B2}) equals
$-\infty$ unless $\mu\in K=\{\mu\in\dom G^\ast\mid\mu_S=0\}$, and on the compact
set $K\ni0$ it equals the finite upper semicontinuous function
$\sum_{z\notin S}\mu_zp_z-G^\ast(\mu)$. This gives attainment, the location of
maximizers, and $-G^\ast(0)\leq\hat G(p)<+\infty$. Monotonicity and lower
semicontinuity follow from Lemma \ref{app:lemma-lsc-scalar-prod}, as a supremum of
nondecreasing continuous functions is nondecreasing and lower semicontinuous.

To show the continuity of $\hat G$ it remains to show that it is upper semicontinuous. Let $p^n\to p$ and $S_n=\{z\mid p^n_z=-\infty\}$; then
$S_n\subseteq S$ for $n$ large. Along a subsequence realizing
$\limsup_n\hat G(p^n)$, take maximizers $\mu^n\to\mu\in\dom G^\ast$ with
$\mu^n_{S_n}=0$. Since $\dom G^\ast$ is bounded and $p^n_z\to p_z\in\R$ for $z\notin S$, there is
$C\in\R$ with $\sum_{z\notin S}\mu^n_zp^n_z\leq C$ for all $n$. Moreover, for $n$
large, $p^n_z<0$ for every $z\in S$. Hence, for $n$ large and every
$z\in S\setminus S_n$,
\begin{equation*}
0\;\geq\;\mu^n_zp^n_z
\;\geq\;\sum_{z'\in S\setminus S_n}\mu^n_{z'}p^n_{z'}
\;=\;\hat G(p^n)+G^\ast(\mu^n)-\sum_{z'\notin S}\mu^n_{z'}p^n_{z'}
\;\geq\;-G^\ast(0)-G(0)-C,
\end{equation*}
so that
\begin{equation*}
0\;\leq\;\mu^n_z\;\leq\;\frac{G^\ast(0)+G(0)+C}{|p^n_z|}
\;\xrightarrow[n\to\infty]{}\;0,
\end{equation*}
since $p^n_z\to-\infty$. As $\mu^n_z=0$ for $z\in S_n$, we obtain $\mu^n_z\to0$
for every $z\in S$, that is, $\mu_z=0$ on $S$ and $\mu\in K$. Dropping these nonpositive terms and using the lower
semicontinuity of $G^\ast$,
$\limsup_n\hat G(p^n)\leq\sum_{z\notin S}\mu_zp_z-G^\ast(\mu)=\langle\mu,p\rangle_--G^\ast(\mu)\leq\hat G(p)$.

Finally, let $p,p'\in\underline\R^\E$ and set
\[
p^n=p\vee(-n\mathbf 1),
\qquad
{p'}^n=p'\vee(-n\mathbf 1).
\]
\noindent The vectors $p^n$ and ${p'}^n$ belong to $\R^\E$, so the
submodularity of $G$ gives
\[
\hat G(p^n\wedge {p'}^n)+\hat G(p^n\vee {p'}^n)
\leq \hat G(p^n)+\hat G({p'}^n).
\]
\noindent Moreover,
\[
p^n\wedge {p'}^n=(p\wedge p')\vee(-n\mathbf 1),
\qquad
p^n\vee {p'}^n=(p\vee p')\vee(-n\mathbf 1).
\]
\noindent Letting $n\to\infty$ and using the continuity of $\hat G$ proves
its submodularity on $\underline\R^\E$.
\end{proof}

From now on we write $G$ for $\hat G$, the canonical continuous extension of the welfare function in Definition~\ref{def:welfare-function}. In accordance with (\ref{eq:choice-correspondence}), for $p\in\underline\R^\E$,
\begin{equation}
C^G(p)=\argmax_{\mu}\{\langle\mu,p\rangle_--G^\ast(\mu)\}.
\label{eq:B4}
\end{equation}
By the Fenchel--Young equality, this is the convex subdifferential when
$p\in\R^\E$. For $\bar\mu\in\R_+^\E$ we write the objective of (\ref{eq:multiplier-set}) with the extended
pairing:
\begin{equation}
T^G(\alpha,\bar\mu)=\argmin_{\tau\in\overline{\R}_+^\E}
\{G(\alpha-\tau)-\langle\bar\mu,\alpha-\tau\rangle_-\}.
\label{eq:B5}
\end{equation}
With the convention $0\cdot(+\infty)=0$, one has
$-\langle\bar\mu,\alpha-\tau\rangle_-=\bar\mu^\top\tau-\bar\mu^\top\alpha$ in
$(-\infty,+\infty]$ for every $\tau$, so (\ref{eq:B5}) agrees with (\ref{eq:multiplier-set}). The
objective in (\ref{eq:B5}) is finite if and only if $\bar\mu_z=0$ whenever
$\tau_z=+\infty$.

\begin{lemma}\label{app:lemma-L-lsc}
The map $(\bar\mu,\tau)\mapsto G(\alpha-\tau)-\langle\bar\mu,\alpha-\tau\rangle_-$
is lower semicontinuous on $\R_+^\E\times\overline{\R}_+^\E$, with values in
$(-\infty,+\infty]$. Hence $T^G(\alpha,\bar\mu)$ is nonempty and compact, and
the minimum in (\ref{eq:B5}) is finite.
\end{lemma}

\begin{proof}
The first term is continuous by Proposition \ref{app:prop-hatG}, and the second is lower
semicontinuous by Lemma \ref{app:lemma-lsc-scalar-prod}. The objective is bounded
below by $-G^\ast(0)-\bar\mu^\top\alpha$, since
$\langle\bar\mu,\alpha-\tau\rangle_-\leq\bar\mu^\top\alpha$, and equals
$G(\alpha)-\bar\mu^\top\alpha$ at $\tau=0$. Its minimum over the compact space
$\overline{\R}_+^\E$ is therefore attained and finite, and the set of minimizers
is closed.
\end{proof}

\begin{lemma}\label{app:lemma-duality}
For $\alpha\in\R^\E$ and $\bar\mu\geq0$,
\begin{equation}
\max_{\substack{\mu\in \R_+^\E\\ \mu\leq \bar\mu}}
\{\mu^\top\alpha-G^\ast(\mu)\}-\bar\mu^\top\alpha
=\min_{\tau\in\overline{\R}_+^{\E}}\{G(\alpha-\tau)-\langle\bar\mu,\alpha-\tau\rangle_-\}
=\inf_{\tau\in\R_+^\E}\{G(\alpha-\tau)-\bar\mu^\top(\alpha-\tau)\}.
\label{eq:B6}
\end{equation}
Moreover, if $\mu\in \mathcal Q^G(\alpha,\bar\mu)$, then $\tau\in T^G(\alpha,\bar\mu)$ if
and only if
\begin{equation}
\mu\in C^G(\alpha-\tau)
\qquad\text{and}\qquad
\tau_z=0\ \text{whenever}\ \mu_z<\bar\mu_z.
\label{eq:B7}
\end{equation}
\end{lemma}

\begin{proof}
The minimum is equal to the infimum by Lemma \ref{app:lemma-lsc-scalar-prod} and Proposition \ref{app:prop-hatG}. The function
$f(\mu,\tau)=\mu^\top\alpha-G^\ast(\mu)+(\bar\mu-\mu)^\top\tau$ on
$\dom G^\ast\times\R_+^\E$ is concave and upper semicontinuous in $\mu$, which
ranges over a convex compact set, and affine in $\tau$. By
\citeasnoun{Sion1958} and (\ref{eq:Grep}),
\[
\max_{\substack{\mu\in \R_+^\E\\ \mu\leq \bar\mu}}
\{\mu^\top\alpha-G^\ast(\mu)\}=\inf_{\tau\geq0}\max_{\mu}f(\mu,\tau)
=\max_{\mu}\inf_{\tau\geq0}f(\mu,\tau)=\inf_{\tau\in\R_+^\E}\{G(\alpha-\tau)+\bar\mu^\top\tau\},
\]
since $\inf_{\tau\geq0}f(\mu,\tau)=-\infty$ unless $\mu\leq\bar\mu$.

Let $\mu\in \mathcal Q^G(\alpha,\bar\mu)$. By (\ref{eq:B6}),
$\mu^\top\alpha-G^\ast(\mu)-\bar\mu^\top\alpha$ equals the finite minimum in
(\ref{eq:B5}). Let $\tau\in\overline{\R}_+^\E$.

\emph{Case 1: the objective in (\ref{eq:B5}) is infinite at $\tau$.} Then
$\tau\notin T^G(\alpha,\bar\mu)$. Moreover (\ref{eq:B7}) fails. Indeed, if
$\mu\in C^G(\alpha-\tau)$, then $\mu_z=0$ whenever $\tau_z=+\infty$ by
Proposition \ref{app:prop-hatG}. If also $\tau_z=0$ whenever $\mu_z<\bar\mu_z$, then
$\bar\mu_z=\mu_z=0$ whenever $\tau_z=+\infty$, which makes the objective finite.

\emph{Case 2: the objective in (\ref{eq:B5}) is finite at $\tau$.} Then
$\bar\mu_z=0$, and hence $\mu_z=0$, whenever $\tau_z=+\infty$. Writing
$F=\{z\mid\tau_z<+\infty\}$, all the following quantities are finite, and
\begin{equation}
\begin{split}
G(\alpha-\tau)-\langle\bar\mu,\alpha-\tau\rangle_-
-\{\mu^\top\alpha-G^\ast(\mu)-\bar\mu^\top\alpha\}
={}&[G(\alpha-\tau)-\langle\mu,\alpha-\tau\rangle_-+G^\ast(\mu)]\\
&+\sum_{z\in F}(\bar\mu_z-\mu_z)\tau_z.
\end{split}
\label{eq:B8}
\end{equation}
By (\ref{eq:B6}), the left-hand side is nonnegative and vanishes exactly when
$\tau\in T^G(\alpha,\bar\mu)$. On the right-hand side, the bracket is
nonnegative by (\ref{eq:B2}) and vanishes exactly when
$\mu\in C^G(\alpha-\tau)$. The sum is nonnegative, since $\mu\leq\bar\mu$
and $\tau\geq0$. It vanishes exactly when $\tau_z=0$ for every $z\in F$ with
$\mu_z<\bar\mu_z$. Since $\bar\mu_z=\mu_z$ off $F$, this is the second
condition in (\ref{eq:B7}). Hence $\tau\in T^G(\alpha,\bar\mu)$ if and only if
(\ref{eq:B7}) holds.
\end{proof}
Because constraint-induced waiting times need not be unique, the
deferred-acceptance algorithm selects the least one. The following lemma
records the needed properties of the multiplier set and its least element.
\begin{lemma}[Least multipliers]\label{app:lemma-least}
Let $\alpha\in\R^\E$ and $0\leq\bar\mu\leq\bar\mu'$. Then
\begin{enumerate}
\item[(a)] $T^G(\alpha,\bar\mu')\leq_P T^G(\alpha,\bar\mu)$; in particular
$T^G(\alpha,\bar\mu)$ is a compact sublattice containing
$\inf T^G(\alpha,\bar\mu)$;
\item[(b)] $\inf T^G(\alpha,\bar\mu)\geq \inf T^G(\alpha,\bar\mu')$;
\item[(c)] $\inf T^G(\alpha,\mu)=\inf T^G(\alpha,\bar\mu)$ for every
$\mu\in \mathcal Q^G(\alpha,\bar\mu)$.
\end{enumerate}
\end{lemma}

\begin{proof}
(a) The definition of the P-order and the proof of
Lemma~\ref{lem:ordered-minimizers}(i) extend verbatim to extended-real-valued
functions on any lattice, in particular to the product lattice
$\overline\R_+^\E$. It therefore suffices to show that the objectives of
(\ref{eq:B5}) at $\bar\mu'$ and at $\bar\mu$ are ordered
in the P-order on $\overline{\R}_+^\E$. Since $\tau\mapsto\alpha-\tau$ exchanges $\wedge$
and $\vee$, this follows by adding the submodularity of $G$
(Proposition \ref{app:prop-hatG}) and (\ref{eq:B3}) with $p=\alpha-\tau$,
$p'=\alpha-\tau'$. Taking $\bar\mu=\bar\mu'$ shows that $T^G(\alpha,\bar\mu)$ is
a sublattice. For each $z$ pick $\tau^{(z)}\in T^G(\alpha,\bar\mu)$ minimizing
$\tau_z$ over this compact set; then $\bigwedge_z\tau^{(z)}\in T^G(\alpha,\bar\mu)$
is the coordinatewise infimum \cite{Veinott1989,Topkis1998}.

(b) With $t=\inf T^G(\alpha,\bar\mu)$ and $t'=\inf T^G(\alpha,\bar\mu')$, (a)
gives $t'\wedge t\in T^G(\alpha,\bar\mu')$, hence $t'\leq t'\wedge t\leq t$
\cite{Topkis1998,MilgromNShannon}.

(c) Since $\mu\leq\bar\mu$ is feasible and optimal, $\mu\in \mathcal Q^G(\alpha,\mu)$. If
$\tau\in T^G(\alpha,\bar\mu)$, Lemma \ref{app:lemma-duality} gives
$\mu\in C^G(\alpha-\tau)$, and the second condition of (\ref{eq:B7}) at
capacity $\mu$ is vacuous, so $\tau\in T^G(\alpha,\mu)$. Hence
$\inf T^G(\alpha,\mu)\leq\inf T^G(\alpha,\bar\mu)$, and (b) gives the reverse.
\end{proof}
Finally, we establish closed-graph and continuity properties for constrained
choices and waiting times.
\begin{lemma}\label{app:lemma-closed-graph}
Let $\bar\mu^k\to\bar\mu$ in $\R_+^\E$.
\begin{enumerate}
\item[(a)] If $\mu^k\in \mathcal Q^G(\alpha,\bar\mu^k)$ and $\mu^k\to\mu$, then
$\mu\in \mathcal Q^G(\alpha,\bar\mu)$.
\item[(b)] If $\tau^k\in T^G(\alpha,\bar\mu^k)$ and $\tau^k\to\tau$ in
$\overline{\R}_+^\E$, then $\tau\in T^G(\alpha,\bar\mu)$.
\item[(c)] If $(\bar\mu^k)$ is nonincreasing, then $\inf T^G(\alpha,\bar\mu^k)$
is nondecreasing and converges to $\inf T^G(\alpha,\bar\mu)$.
\end{enumerate}
\end{lemma}

\begin{proof}
(a) Clearly $\mu\leq \bar\mu$. Let $\nu\in\dom G^\ast$,
$\nu\leq\bar\mu$, and
$\lambda_k=\min(\{1\}\cup\{\bar\mu^k_z/\nu_z\mid\nu_z>0\})\to1$. Then
$\lambda_k\nu$ is feasible at $\bar\mu^k$ and, by convexity and
$0\in\dom G^\ast$, $G^\ast(\lambda_k\nu)\leq\lambda_kG^\ast(\nu)+(1-\lambda_k)G^\ast(0)$.
Passing to the limit in the optimality of $\mu^k$ against $\lambda_k\nu$, using
the lower semicontinuity of $G^\ast$, gives
$\mu^\top\alpha-G^\ast(\mu)\geq\nu^\top\alpha-G^\ast(\nu)$.

(b) For every $s\in\R_+^\E$, optimality of $\tau^k$ and the joint lower
semicontinuity of Lemma \ref{app:lemma-L-lsc} give
\[
G(\alpha-\tau)-\langle\bar\mu,\alpha-\tau\rangle_-
\leq\liminf_k\{G(\alpha-s)-(\bar\mu^k)^\top(\alpha-s)\}
=G(\alpha-s)-\bar\mu^\top(\alpha-s).
\]
Taking the infimum over $s$ and using (\ref{eq:B6}) shows that $\tau$ is a
minimizer.

(c) Since $\bar\mu\leq\bar\mu^{k+1}\leq\bar\mu^k$, Lemma \ref{app:lemma-least}(b)
shows that $\inf T^G(\alpha,\bar\mu^k)$ is nondecreasing and bounded by
$\inf T^G(\alpha,\bar\mu)$. Its limit belongs to $T^G(\alpha,\bar\mu)$ by (b),
hence equals $\inf T^G(\alpha,\bar\mu)$.
\end{proof}

\subsection{Well-definedness of the algorithm}\label{app:well-definedness}
Recall that the algorithm starts from $\mu^{A,0}=n^G\wedge n^H$ and
$\mu^{T,-1}=0$ and that, for $k=0,1,\dots$, it selects
\begin{align}
\mu^{P,k}&\in\argmax_{\mu^{T,k-1}\leq\mu\leq\mu^{A,k}}
\{\mu^\top\alpha-G^\ast(\mu)\},
\tag{$\mathrm P_k$}\label{eq:Pk}\\
\mu^{T,k}&\in\argmax_{\ 0\leq\mu\leq\mu^{P,k}}
\{\mu^\top\gamma-H^\ast(\mu)\}=\mathcal Q^H(\gamma,\mu^{P,k}),
\tag{$\mathrm D_k$}\label{eq:Dk}\\
\mu^{A,k+1}&=\mu^{A,k}-(\mu^{P,k}-\mu^{T,k}).
\tag{$\mathrm U_k$}\label{eq:Uk}
\end{align}

We show by induction that (\ref{eq:Pk}) and (\ref{eq:Dk}) have solutions at
every round, that every solution of (\ref{eq:Pk}) lies in
$\mathcal Q^G(\alpha,\mu^{A,k})$, and that
\begin{equation}
0\leq\mu^{T,k}\leq\mu^{P,k}\leq\mu^{A,k},
\qquad
0\leq\mu^{A,k+1}\leq\mu^{A,k}.
\label{eq:B9}
\end{equation}
At $k=0$, since $\mu^{T,-1}=0$ and $\dom G^\ast\subseteq\R_+^\E$, the solution
set of $(\mathrm P_0)$ is $\mathcal Q^G(\alpha,\mu^{A,0})$. Once
$\mu^{P,k}\leq\mu^{A,k}$ is selected, the solution set $\mathcal Q^H(\gamma,\mu^{P,k})$
of (\ref{eq:Dk}) is nonempty and contained in $[0,\mu^{P,k}]$. Then
(\ref{eq:Uk}) gives
$\mu^{T,k}\leq\mu^{A,k+1}=\mu^{A,k}-(\mu^{P,k}-\mu^{T,k})\leq\mu^{A,k}$, which
is (\ref{eq:B9}).

For the inductive step, write
$\mathcal R^G(\alpha,\bar\mu)=\{\bar\mu-\nu\mid\nu\in \mathcal Q^G(\alpha,\bar\mu)\}$ for the set of
rejected quantities. Since $G$ is submodular, $G^\ast$ is exchangeable, and
the capacity-induced response of Theorem~\ref{thm:rejection-comparative-statics} gives
\[
\bar\mu\leq\bar\mu'
\;\Longrightarrow\;
\mathcal R^G(\alpha,\bar\mu)\leq_Q \mathcal R^G(\alpha,\bar\mu').
\]
Applied with $\bar\mu=\mu^{A,k+1}\leq\bar\mu'=\mu^{A,k}$, this yields
$\mathcal R^G(\alpha,\mu^{A,k+1})\leq_Q \mathcal R^G(\alpha,\mu^{A,k})$. Both sets are nonempty.
By Proposition~\ref{prop:orders-imply-weak-order}, the Q-order between nonempty sets implies the weak set
order: every element of $\mathcal R^G(\alpha,\mu^{A,k})$ lies above some element of
$\mathcal R^G(\alpha,\mu^{A,k+1})$. Applied to
$\mu^{A,k}-\mu^{P,k}\in \mathcal R^G(\alpha,\mu^{A,k})$, this provides
$\widehat\mu\in \mathcal Q^G(\alpha,\mu^{A,k+1})$ with
\[
\mu^{A,k+1}-\widehat\mu\leq\mu^{A,k}-\mu^{P,k},
\qquad\text{that is,}\qquad
\widehat\mu\geq\mu^{A,k+1}-(\mu^{A,k}-\mu^{P,k})=\mu^{T,k}
\]
by (\ref{eq:Uk}). Hence $\widehat\mu$ is feasible for $(\mathrm P_{k+1})$.
It is also optimal over the larger set
$\{\mu\in\dom G^\ast\mid\mu\leq\mu^{A,k+1}\}$. So the solution set of
$(\mathrm P_{k+1})$ is
$\mathcal Q^G(\alpha,\mu^{A,k+1})\cap\{\mu\geq\mu^{T,k}\}\ni\widehat\mu$. Consequently
the algorithm is well defined and, for every $k$,
\begin{equation}
\mu^{P,k}\in \mathcal Q^G(\alpha,\mu^{A,k}),
\qquad
\mu^{T,k}\in \mathcal Q^H(\gamma,\mu^{P,k}),
\qquad
\mu^{T,k}\leq\mu^{P,k+1}.
\label{eq:B10}
\end{equation}
\subsection{Monotonicity and limiting optimality}

For every $k$, recall that we use the least multipliers
\[
\tau^{P,k}=\inf T^G(\alpha,\mu^{A,k}),
\qquad
\tau^{T,k}=\inf T^H(\gamma,\mu^{P,k}).
\]
By Lemma \ref{app:lemma-least}(a), $T^G(\alpha,\mu^{A,k})$ and
$T^H(\gamma,\mu^{P,k})$ are nonempty compact sublattices that contain their
coordinatewise infima. Hence $\tau^{P,k}\in T^G(\alpha,\mu^{A,k})$ and
$\tau^{T,k}\in T^H(\gamma,\mu^{P,k})$ are well defined.

\begin{lemma}\label{app:lemma-monotone}
$(\tau^{P,k})$ is nondecreasing and $(\tau^{T,k})$ is nonincreasing.
\end{lemma}

\begin{proof}
The first claim follows from $\mu^{A,k+1}\leq\mu^{A,k}$ and
Lemma \ref{app:lemma-least}(b). For the second, (\ref{eq:B10}) and
Lemma \ref{app:lemma-least}(c),(b) give
$\tau^{T,k}=\inf T^H(\gamma,\mu^{T,k})\geq\inf T^H(\gamma,\mu^{P,k+1})=\tau^{T,k+1}$.
\end{proof}

\begin{lemma}\label{app:lemma-limits}
$\mu^{A,k}\to\mu^A$, and $\tau^{P,k}\to\tau^P$, $\tau^{T,k}\to\tau^T$ in
$\overline{\R}_+^\E$. The proposal sequence has convergent subsequences, and
along any subsequence with $\mu^{P,k_\ell}\to\mu$ one has
$\mu^{T,k_\ell}\to\mu$ and
\begin{equation}
\mu\in \mathcal Q^G(\alpha,\mu^A)\cap \mathcal Q^H(\gamma,\mu),
\qquad
\mu\in C^G(\alpha-\tau^P)\cap C^H(\gamma-\tau^T),
\qquad
\tau^P_z=0\ \text{whenever}\ \mu_z<\mu^A_z.
\label{eq:B11}
\end{equation}
\end{lemma}

\begin{proof}
Monotone sequences converge in $\R_+$ and $\overline{\R}_+$, by (\ref{eq:B9}) and
Lemma \ref{app:lemma-monotone}. By the update rule (\ref{eq:Uk}),
$\mu^{P,k}-\mu^{T,k}=\mu^{A,k}-\mu^{A,k+1}\to0$.
Consider a subsequence, not relabeled, along which $\mu^{P,k}\to\mu$. Then
$\mu^{T,k}\to\mu$ along the same subsequence. By (\ref{eq:B10}), along
this subsequence,
\[
\mu^{P,k}\in \mathcal Q^G(\alpha,\mu^{A,k}),\quad
\mu^{A,k}\to\mu^A,\quad
\mu^{P,k}\to\mu,
\qquad
\mu^{T,k}\in \mathcal Q^H(\gamma,\mu^{P,k}),\quad
\mu^{P,k}\to\mu,\quad
\mu^{T,k}\to\mu.
\]
Applying Lemma \ref{app:lemma-closed-graph}(a) to $G$ with capacities $\mu^{A,k}$,
and to $H$ with capacities $\mu^{P,k}$, yields
\[
\mu\in \mathcal Q^G(\alpha,\mu^A),
\qquad
\mu\in \mathcal Q^H(\gamma,\mu).
\]
For the dual limits, the capacities $(\mu^{A,k})$ are nonincreasing with limit
$\mu^A$, and $\tau^{P,k}=\inf T^G(\alpha,\mu^{A,k})$.
Lemma \ref{app:lemma-closed-graph}(c) therefore gives
\[
\tau^P=\lim_{k\to\infty}\inf T^G(\alpha,\mu^{A,k})=\inf T^G(\alpha,\mu^A)
\in T^G(\alpha,\mu^A).
\]
Along the subsequence,
\[
\tau^{T,k}\in T^H(\gamma,\mu^{P,k}),
\qquad
\mu^{P,k}\to\mu,
\qquad
\tau^{T,k}\to\tau^T\ \text{in}\ \overline{\R}_+^\E,
\]
so Lemma \ref{app:lemma-closed-graph}(b) gives $\tau^T\in T^H(\gamma,\mu)$.
Finally, Lemma \ref{app:lemma-duality} applies to the pairs
\[
\mu\in \mathcal Q^G(\alpha,\mu^A),\ \tau^P\in T^G(\alpha,\mu^A)
\qquad\text{and}\qquad
\mu\in \mathcal Q^H(\gamma,\mu),\ \tau^T\in T^H(\gamma,\mu).
\]
By (\ref{eq:B7}), it yields
\[
\mu\in C^G(\alpha-\tau^P),
\qquad
\tau^P_z=0\ \text{whenever}\ \mu_z<\mu^A_z,
\qquad
\mu\in C^H(\gamma-\tau^T),
\]
which, together with the primal memberships, is (\ref{eq:B11}).
\end{proof}

\subsection{Normalization and aggregate equilibrium}

Let $\mu,\tau^P,\tau^T$ be as in Lemma \ref{app:lemma-limits}, and recall the normalized multipliers
$\tau^\alpha,\tau^\gamma\in\overline{\R}_+^\E$ by
\[
\tau_z^\alpha=
\begin{cases}
\tau_z^P,&\mu_z<n_z^G,\\
0,&\mu_z=n_z^G,
\end{cases}
\qquad
\tau_z^\gamma=
\begin{cases}
\tau_z^T,&\mu_z<n_z^H,\\
0,&\mu_z=n_z^H.
\end{cases}
\]
This renormalization will still yield waiting times associated with the choice $\mu$ while ensuring that both sides of the market are not waiting for a given pairing.
\begin{lemma}\label{app:lemma-normalization}
$\mu\in C^G(\alpha-\tau^\alpha)\cap C^H(\gamma-\tau^\gamma)$.
\end{lemma}

\begin{proof}
We treat $G$, the argument for $H$ is similar. Let $I=\{z\mid\mu_z=n^G_z>0\}$. On $\{z\mid\mu_z=n^G_z=0\}$ every
$\nu\in\dom G^\ast$ vanishes, so changing $\tau$ there does not affect
$\langle\nu,\alpha-\tau\rangle_-$. For $z\in I$, $\mu_z>0$ and
$\mu\in C^G(\alpha-\tau^P)$ force $\tau^P_z<+\infty$. Hence, for every
$\nu\in\dom G^\ast$,
\[
\langle\nu,\alpha-\tau^\alpha\rangle_-
=\langle\nu,\alpha-\tau^P\rangle_-+\sum_{z\in I}\nu_z\tau^P_z.
\]
\noindent Since $\nu_z\leq n^G_z=\mu_z$ on $I$, optimality of $\mu$ at
$\tau^P$ yields
\[
\begin{aligned}
&\langle\mu,\alpha-\tau^\alpha\rangle_- -G^\ast(\mu)
-\left(\langle\nu,\alpha-\tau^\alpha\rangle_- -G^\ast(\nu)\right)\\
&\quad =
\langle\mu,\alpha-\tau^P\rangle_- -G^\ast(\mu)
-\left(\langle\nu,\alpha-\tau^P\rangle_- -G^\ast(\nu)\right)
+\sum_{z\in I}(\mu_z-\nu_z)\tau_z^P
\geq0.
\end{aligned}
\]
\noindent Thus $\mu\in C^G(\alpha-\tau^\alpha)$.
\end{proof}
\begin{lemma}\label{app:lemma-no-concession}
$\min\{\tau_z^\alpha,\tau_z^\gamma\}=0$ for every $z\in\E$.
\end{lemma}

\begin{proof}
Otherwise $\mu_z<n_z^G\wedge n_z^H$, $\tau_z^P>0$ and $\tau_z^T>0$. By
Lemma \ref{app:lemma-monotone}, $\tau_z^{T,k}\geq\tau_z^T>0$ for all $k$. At round
$k$, the disposal step (\ref{eq:Dk}) selects $\mu^{T,k}\in \mathcal Q^H(\gamma,\mu^{P,k})$,
and $\tau^{T,k}\in T^H(\gamma,\mu^{P,k})$. The second condition of
(\ref{eq:B7}), applied to $H$ with capacity $\mu^{P,k}$, therefore forces
$\mu_z^{T,k}=\mu_z^{P,k}$ for all $k$. By the update rule (\ref{eq:Uk}), no
mass is then removed in coordinate $z$:
$\mu^{A,k+1}_z=\mu^{A,k}_z-(\mu^{P,k}_z-\mu^{T,k}_z)=\mu^{A,k}_z$ for all $k$.
Since the algorithm starts from $\mu^{A,0}=n^G\wedge n^H$, we get
$\mu^A_z=\mu^{A,0}_z=n_z^G\wedge n_z^H$. Since $\tau_z^P>0$, (\ref{eq:B11})
gives $\mu_z=\mu_z^A=n_z^G\wedge n_z^H$, a contradiction.
\end{proof}

\begin{proof}[Proof of Theorem~\ref{thm:aggregate-da}]
Items 1--4 follow from Section~\ref{app:well-definedness} and
Lemmas \ref{app:lemma-monotone} and \ref{app:lemma-limits}. For item 5, recall that
$(\mu,U,V)\in\R_+^\E\times\underline\R^\E\times\underline\R^\E$ is an aggregate
equilibrium matching if
\[
\mu\in C^G(U)\cap C^H(V)
\qquad\text{and}\qquad
\max\{U_z-\alpha_z,V_z-\gamma_z\}=0
\quad\text{for every }z\in\E.
\]
Let $U=\alpha-\tau^\alpha$ and $V=\gamma-\tau^\gamma$ in $\underline\R^\E$.
Lemma \ref{app:lemma-normalization} gives $\mu\in C^G(U)\cap C^H(V)$.
Since $\tau^\alpha,\tau^\gamma\geq0$, Lemma \ref{app:lemma-no-concession} gives
\[
\max\{U_z-\alpha_z,V_z-\gamma_z\}
=-\min\{\tau^\alpha_z,\tau^\gamma_z\}=0
\quad\text{for every }z\in\E.
\]
Thus $(\mu,U,V)$ is an aggregate equilibrium matching. Corollary~\ref{cor:burn-only-existence} follows
from item 4, which guarantees that a convergent proposal subsequence exists.
\end{proof}

\end{document}